\documentclass[pdflatex,sn-mathphys-num,rstransa]{sn-jnl}

\usepackage[T1]{fontenc}
\usepackage{graphicx}%
\usepackage{multirow}%
\usepackage{amsmath,amssymb,amsfonts}%
\usepackage{amsthm}%
\usepackage{mathrsfs}%
\usepackage[title]{appendix}%
\usepackage{xcolor}%
\usepackage{textcomp}%
\usepackage{manyfoot}%
\usepackage{booktabs}%
\usepackage{algorithm}%
\usepackage{algorithmicx}%
\usepackage{algpseudocode}%
\usepackage{listings}%
\usepackage{quiver}

\newcommand{\tl}[1]{\tilde{#1}}
\newcommand{\lift}[1]{\uparrow_{#1}}
\newcommand{\down}[1]{\downarrow_{#1}}

\newcommand{\aaa}[1]{\left \langle#1\right\rangle}

\newcommand{\pp}[2]{\frac{\partial #1}{\partial #2}}
\newcommand{\bpp}[2]{\frac{\boldsymbol{\partial} #1}{\boldsymbol{\partial} #2}}
\newcommand{\dd}[2]{\frac{\mathrm{d} #1}{\mathrm {d} #2}}

\newcommand{\conn}[1]{\overset{#1}{\nabla}}

\newcommand{\Conn}[3]{\overset{#1}{\nabla}_{#2}{#3}}

\newcommand{\bz}{\boldsymbol{z}}

\newcommand{\bv}{\boldsymbol{v}}

\newcommand{\bGamma}{\boldsymbol{\Gamma}}

\newcommand{\bZ}{\boldsymbol{Z}}
\newcommand{\bV}{\boldsymbol{V}}
\newcommand{\bU}{\boldsymbol{U}}
\newcommand{\bW}{\boldsymbol{W}}
\newcommand{\bw}{\boldsymbol{w}}

\newcommand{\bY}{\boldsymbol{Y}}
\newcommand{\bF}{\boldsymbol{F}}

\newcommand{\bH}{\boldsymbol{H}}
\newcommand{\blamb}{\boldsymbol{\lambda}}

\newcommand{\bX}{\boldsymbol{X}}

\newcommand{\bzero}{\boldsymbol{0}}

\newcommand{\bq}{\boldsymbol{q}}

\newcommand{\bh}{\boldsymbol{h}}

\newcommand{\bdd}{\boldsymbol{d}}

\newcommand{\om}{\omega}

\newcommand{\mbG}{\mathbb{G}}
\newcommand{\mbF}{\mathbb{F}}
\newcommand{\mbA}{\mathbb{A}}
\newcommand{\mbS}{\mathbb{S}}

\newcommand{\mbR}{\mathbb{R}}

\newcommand{\mbO}{\mathbb{O}}
\newcommand{\mbQ}{\mathbb{Q}}

\newcommand{\mbI}{\mathbb{I}}

\newcommand{\mcD}{\mathcal{D}}

\newcommand{\mcP}{\mathcal{P}}
\newcommand{\mcQ}{\mathcal{Q}}
\newcommand{\mcA}{\mathcal{A}}
\newcommand{\mcR}{\mathcal{R}}
\newcommand{\mcH}{\mathcal{H}}
\newcommand{\mcI}{\mathcal{I}}

\newcommand{\mcX}{\mathcal{X}}

\newcommand{\mcL}{\mathcal{L}}

\newcommand{\mcV}{\mathcal{V}}

\newcommand{\eps}{\varepsilon}

\newcommand{\etc}{\textit{etc}}

\newcommand{\ie}{\textit{i.e.}}

\newtheorem{theorem}{Theorem} 
\newtheorem{corollary}{Corollary}[theorem]
\newtheorem{lemma}{Lemma}
\newtheorem{remark}{Remark}
\newtheorem{prop}{Proposition}

\newcommand{\rem}[1]{}

\newcommand{\mbid}{\mbox{id}}

\newcommand{\ver}{\mbox{ver}}
\newcommand{\hor}{\mbox{hor}}
\newcommand{\sspan}{\mbox{span}}
\usepackage[utf8]{inputenc}

\begin{document}

\title[Affine Connection Approach to the Realization of Nonholonomic Constraints by Strong Friction Forces]{Affine Connection Approach to the Realization of Nonholonomic Constraints by Strong Friction Forces}


\author{\fnm{Vaughn} \sur{Gzenda}}\email{vaughngzenda@cmail.carleton.ca}

\author*{\fnm{Robin} \sur{Chhabra*}}\email{robin.chhabra@carleton.ca}

\affil{\orgdiv{Department of Mechanical and Aerospace Engineering}, \orgname{Carleton University}, \orgaddress{\street{1125 Colonel By Dr}, \city{Ottawa}, \postcode{K1S 5B6}, \state{Ontario}, \country{Canada}}}


\abstract{   In this paper, we study an affine connection approach to realizing nonholonomic mechanical systems mediated by viscous friction forces with large coefficients, viewed as a singular perturbation of the nonholonomic system. We show that the associated slow manifold is represented coordinate-free as the image of a section over the nonholonomic distribution. We propose a novel invariance condition based on covariant derivatives and prove that this condition is equivalent to the classical invariance condition based on time derivatives. Accordingly, we propose a novel recursive procedure to approximate the slow manifold based on the covariant derivatives of a formal power series expansion of the section. Using this recurrence relation, we derive, up to second order, approximations of the slip velocities residing in the slow manifold, as well as the associated approximated dynamics up to first order. Lastly, we illustrate our approach with a case study of a vertical rolling disk. }




\maketitle
\textbf{Submitted to the Journal of Nonlinear Dynamics}

\section{Introduction}
Rolling without slipping is a constraint that appears in many mechanical systems, such as vehicle systems, and is commonly viewed as an example of an ideal nonholonomic constraint \cite{bloch2004nonholonomic}. Nonholonomic constraints normally refer to point-wise linear constraints on the velocity space at each configuration, which cannot be described by a derivative of a set of constraints on the configuration space \cite{bloch2004nonholonomic}. Hence, such constraints are then viewed as a nonintegrable linear distribution of the tangent bundle to the configuration manifold \cite{bloch2004nonholonomic}. However, the nonholonomic constraint forces are only valid for velocities in the distribution. In practice, nonholonomic constraints in vehicle systems are often violated, resulting in an accumulation of errors for standard nonholonomic model-based control laws. To model the velocities violating the nonholonomic constraints, the ideal nonholonomic constraint forces may be modelled as dissipative friction, whose limit at large coefficients approximates the ideal nonholonomic system \cite{eldering2016realizing}. The large friction coefficients introduce a steep linear relationship between the friction force and the motion of the contact point that is analogous to the linear regime of Bakker-Pacejka magic formula \cite{pacejka1992magic}. Mathematically, the system with large coefficients may be treated as a slow/fast singular perturbation where the dissipative friction forces act on a much faster time scale than the nominal variables resulting in an exponentially attractive invariant manifold. This manifold -- often called the slow manifold -- physically represents small violations of the nonholonomic constraints. 

Nonholonomic mechanical systems form a class of constrained systems whose velocity constraints have a natural differential geometric structure  \cite{bloch2004nonholonomic,bates1993nonholonomic,bloch2009quasivelocities,bloch2010nonholonomic}. Classically, the equations of motion are derived using Lagrange-d'Alembert's (L-d'A) variational principle -- which is an adaptation of Hamilton's variational principle --  whereby the variations are constrained to the nonholonomic distribution and enforced with Lagrange multipliers that do no work on the distribution \cite{bloch2004nonholonomic}. Some examples of nonholonomic constraints include rolling without slipping, no sliding, and conserved angular momenta, among
others \cite{bloch2004nonholonomic,bullo2019geometric}. The kinetic energy of the mechanical system may be interpreted geometrically as a Riemannian metric tensor \cite{petersen2006riemannian}. Equivalent to the Lagrange-d'Alembert principle, Lewis and Bullo \cite{lewis1998affine,lewis2000simple} use Riemannian geometry to model the nonholonomic dynamics with affine connections. An affine connection is a geometric object that enables a notion of differentiation of vector fields on a smooth manifold. A unique affine connection is derived from the metric tensor called the Levi-Civita affine connection that encodes Coriolis, centrifugal, and centripetal accelerations. In \cite{bullo2019geometric}, the nonholonomic constraints may be viewed as the kernel of a projection map onto the complement of the nonholonomic distribution. In the affine connection framework, \cite{bullo2019geometric}, the nonholonomic constraints are then enforced using Lagrange multipliers that are solved for using a covariant derivative of the complementary projection map. As such, Lewis introduces an affine connection called the nonholonomic connection, based on the Levi-Civita connection and the covariant of the orthogonal projection map, modelling the necessary accelerations required to maintain the nonholonomic constraints. The geometric interplay of affine connections and nonintegrable distributions is studied in detail in \cite{lewis1998affine}. Other equivalent approaches to nonholonomic mechanics include symplectic geometry \cite{bates1993nonholonomic,chhabra2014nonholonomic}, Poisson geometry \cite{van1994hamiltonian}, Gauss's Principle of Least Constraint \cite{bloch2004nonholonomic,bullo2019geometric} and Kane's equations \cite{kane1985dynamics}.

Singularly perturbed systems (SPS) manifest in various domains of science and engineering, emerging when a dynamical system relies on scalable parasitic parameters like small time constants, masses, friction, \etc. Accordingly, SPS have natural applications in electric circuits \cite{sannuti1981singular}, chemistry \cite{zagaris2004analysis}, neuroscience \cite{rubin2002geometric}, among others. Slow/fast systems are a particular case of a SPS that involve two distinct time scales, delineated by a singular parameter $\epsilon$. Tikhonov established that when the spectra of the linearized fast dynamics lies in the left half-plane, the flow converges to an invariant (slow) manifold as $\epsilon$ tends to 0 \cite{tikhonov1952systems}. The functional relationship between the slow and fast variables is dictated by an \emph{invariance condition} and a formal power series expansion is a common approach to approximate the slow manifold \cite{verhulst2005methods}. Fenichel examined the persistence of the invariant manifold under small perturbations of $\eps$, utilizing the concept of Normally Hyperbolic Invariant Manifold (NHIM) and under the assumption of compactness in the state space \cite{fenichel1971persistence}. In \cite{fenichel1979geometric}, Fenichel utilizes a geometric approach to investigate asymptotic expansions of the graph representation of the invariant manifold. In the Riemannian geometric context, Eldering extends Fenichel's findings to non-compact state spaces, assuming regularity under bounded geometry conditions \cite{eldering2013normally}. Eldering's work opens up a pathway for the coordinate-independent examination of such systems. If a mechanical system possesses a noncompact configuration space $Q$ and is endowed with a symmetry group $G$, where the group action leads to a compact \emph{shape space} $Q/G$ \cite{marsden1998introduction}, then it follows that $Q$ has bounded geometry \cite{eldering2013normally,eldering2016realizing}. In practical applications, numerous mechanical systems exhibit symmetry groups, as seen in mobile vehicles that yield compact shape spaces \cite{chhabra2014nonholonomic}. Hence, Eldering's geometric singular perturbation techniques \cite{eldering2013normally} can be employed for the analysis of such systems.

The physical realization of nonholonomic mechanical systems by strong viscous friction forces was treated classically by Brendelev and Karapetian who mainly focused on convergence to the ideal nonholonomic dynamics \cite{brendelev1981realization, karapetian1981realizing}. More recently, Eldering used the geometric singular perturbation theory of NHIM's \cite{eldering2013normally} to generalize these convergence results for systems with bounded geometry. Common to all these models, the contact forces are modelled as a Rayleigh dissipation function (RDF) defined by a symmetric $(0,2)$-tensor. Rayleigh dissipation functions
have been used to model surface contact forces following first physical principles \cite{minguzzi2015rayleigh}. To realize the nonholonomic constraints, the RDF is constructed such that its kernel is identified with the nonholonomic distribution ensuring that the contact forces do no work on the distribution, mimicking the L-d'A principle, yet the dynamics are exponentially attractive to the distribution. At the limit of large friction coefficients, Eldering proves the existence of a persistent exponentially attractive NHIM and shows that in a nonstandard frame spanning the distribution and its complement, the slow manifold may be represented as the graph over the nonholonomic variables \cite{eldering2016realizing}. It is shown that this graph may be represented as a formal power series, which is used to approximate the slow manifold up to first order \cite{eldering2016realizing} based on the classical invariance condition.

Antali and Stepan considered the physical realization of the rolling phenomenon of wheels by dry friction forces \cite{antali2019nonsmooth}. Their analysis is based on the crossing dynamics of extended Filippov systems \cite{antali2018sliding}. A Filippov system is described by a piece-wise smooth vector field with discontinuities on a \emph{discontinuity manifold} \cite{filippov2013differential}. Extended Filippov (EF) systems are Filippov systems with codimension-2 discontinuity manifolds and describe Coulomb rigid body dynamics \cite{antali2018sliding}. The EF approach has been used to analyze the spinning and sliding behavior of rigid bodies under contact \cite{antali2022nonsmooth}. Recently, bifurcation analysis of EF systems has also been developed to analyze slip and sliding dynamics of a wheel motion \cite{antali2023bifurcations}. In contrast to the slow-fast realization by viscous friction forces \cite{eldering2016realizing}, the Coulomb friction in \cite{antali2019nonsmooth} is both nonsmooth and nonlinear in velocities. The wheel rolling phenomenon has also been captured as a limit of piece-wise continuous holonomic constraints \cite{ruina1998nonholonomic}. Other historical approaches to the realization of nonholonomic constraints include coupling of an external force to generate nonholonomic dynamics \cite{bloch2008quantization}.

In this paper, we study realizations of nonholonomic systems as slow/fast systems using an affine connection approach inspired by the work of Bullo and Lewis \cite{bullo2019geometric,lewis1998affine,lewis2000simple}. In our geometric approach, we represent the slow manifold coordinate-independent as the image of a section over the nonholonomic distribution. We introduce a novel invariance based on a covariant derivative along the integral curves on the slow manifold. We prove that this novel invariance condition is equivalent to the invariance condition based on time derivatives \cite{eldering2016realizing}. We propose a recursion to approximate the velocities in the slow manifold based on covariant derivatives of geometric quantities. Using the recursion, we derive approximations of equations of motion. Lastly, we study an example system of a vertical rolling disk. The major contributions of this work are summarized as follows:
\begin{enumerate}
    \item formulating a coordinate-free representation of the slow manifold as the image of a nonlinear section over the nonholonomic distribution, 
    \item formulating a novel invariance condition based on covariant derivatives, and proving its equivalence to the classical invariance condition, 
    \item proposing a novel recurrence relation to approximate the slip velocities described by the slow manifold,
    \item using the recurrence relation to derive higher-order perturbations of the equations of motion.
\end{enumerate}

The paper is organized as follows. Section \ref{sec_preliminaries} provides a self-contained summary of the differential geometric structures used in the remainder of the article. This includes a review of Riemannian geometry, affine connection systems, nonholonomic affine connection systems, and the geometry of the second tangent bundle including the procedure to take covariant derivatives of general (nonlinear) tangent bundle maps. In Section \ref{sec_affine_connection_systems}, we derive the equations of motion for a slow fast realization of nonholonomic systems and represent the slow manifold as the image of a section. Using this model, in Section \ref{sec_approximations_slip}, we derive the relationship between the reaction forces and the slow manifold and calculate zeroth, first and second order approximations of the slip velocities, and zeroth and first-order dynamical equations of motion. Lastly, in Section \ref{sec_case_study} we consider our approach with a case study of a vertical rolling disk.


\section{Preliminaries} \label{sec_preliminaries}
This section provides a self-contained review of mathematical notions used throughout the paper. For references on Riemannian geometry, geometric mechanics and related topics see \cite{bloch2004nonholonomic,bullo2019geometric,conlon2001differentiable,lewis2020nonholonomic,lewis2000simple,petersen2006riemannian} and references therein. 
\subsection{Riemannian Geometry and Affine Connection Systems}

Let $Q$ be an $n$-dimensional smooth manifold representing the configuration space of a simple mechanical system. Let $TQ = \bigsqcup_{q\in Q}T_qQ $ and $T^*Q = \bigsqcup_{q\in Q}T_q^*Q$ be the disjoint union of tangent and cotangent spaces called the tangent and cotangent bundle, respectively. Let $\pi_{TQ}:TQ \rightarrow Q$ be the bundle projection map for $TQ$ and let $\mcX TQ$ be the set be the set of vector fields on $Q$. In local coordinates $\left(U,\bq = (q^1,\dots,q^n)\right)$, we denote the coordinate induced basis by $\{\bpp{}{q^1},\dots, \bpp{}{q^n}\}$ for $T_{\bq}Q$ and $\{\bdd q^{1},\dots,\bdd q^{n}\}$ for $T^{*}_{\bq}Q$. The natural pairing $\aaa{.|.}:T^{*}_{\bq}Q\times T_{\bq}Q \rightarrow \mbR$ between vectors and covectors in the tangent and cotangent spaces, respectively, is defined by 
\begin{align}
    \aaa{\om_{\bq}|\bX_{\bq}} = \sum^{n}_{i=1}\om_{i} X^{i}, \quad \om_{\bq} \in T^{*}_{\bq }Q, \ \bX_{\bq} \in T_{\bq}Q.
\end{align} where $\bX_{\bq} = \sum_{i}X^{i}\bpp{}{q^{i}}$ and $\om_{\bq} = \sum_{i}\om_{i}\bdd q^{i}$ in the canonical basis of the tangent and cotangent bundles, having components $X^1,\dots,X^n \in \mbR$ and $\om_1,\dots,\om_n \in \mbR$, respectively.
 
Let $\mbG_{\bq}:T_{\bq}Q \times T_{\bq}Q \rightarrow \mbR$ be a Riemannian metric on $Q$ defining the kinetic energy of the system \cite{petersen2006riemannian}. For each $\bq \in Q$ the metric $\mbG$ induces two isomorphisms between the tangent and cotangent bundles given point-wise by 
\begin{align}
    &\begin{aligned}
    &\sharp_{\mbG}:T^{*}_{\bq}Q\rightarrow T_{\bq}Q\\
    &\om_{\bq} \mapsto \mbG(\om_{\bq}^{\sharp_{\mbG}},\bY_{\bq}) = \aaa{\om_{\bq}|\bY_{\bq}},\quad \forall \bY_{\bq} \in T_{\bq}Q
    \end{aligned}\\
    &\begin{aligned}
    &\flat_{\mbG}:T_{\bq}Q\rightarrow T^{*}_{\bq}Q\\
    &\bX_{\bq} \mapsto \aaa{\bX_{q}^{\flat_{\mbG}}|\bY_{\bq}} = \mbG_{\bq}(\bX_{\bq},\bY_{\bq}),\quad \forall \bY_{\bq} \in T_{\bq}Q
    \end{aligned}
\end{align} called the sharp and flat isomorphisms, respectively \cite{bullo2019geometric,petersen2006riemannian}. In the canonical local coordinates for $TQ$ and $T^{*}Q$, the metric tensor takes the matrix form such that by abuse of notation we denote the matrix by $\mbG_{\bq}$. As the result, the sharp and flat isomorphisms become $\om_{\bq}^{\sharp_{\mbG}} = \mbG_{\bq}^{-1}\om_{\bq}$ and $\bX^{\flat_{\mbG}}_{\bq} = \mbG_{\bq}\bX_{\bq}$, where $\om_{\bq}$ and $\bX_{\bq}$ are expressed in the canonical bases. 

The Levi-Civita connection $\Conn{}{}{}$, is the unique affine connection associated with $\mbG$ and is derived from the assumption that the connection is divergentless and torsion-free \cite{bullo2019geometric, petersen2006riemannian}. The $\nabla_{\bX}\bY$ be the covariant derivative of a vector field $\bY$ along a vector field $\bX$ with respect to $\nabla$ is given by 
\begin{align}\label{covariant_d_vector}
    \Conn{}{\bX}{\bY} = \left(\pp{Y^i}{q^j}X^j + \Gamma^i_{jk}X^jY^k \right)\bpp{}{q^i},
\end{align} where the Levi-Civita connection coefficients -- also called the Chistoffel symbols-- $\Gamma^i_{jk}$ in the coordinate-induced basis is expressed as 
\begin{align}
    \Gamma^i_{jk} = \frac{1}{2}\mbG^{il}\left(\pp{\mbG_{jl}}{q^k} + \pp{\mbG_{kl}}{q^j} - \pp{\mbG_{jk}}{q^l} \right).
\end{align} The Christoffel symbols have the physical meaning for a simple mechanical system of Coriolis, centripetal and centrifugal accelerations \cite{bullo2019geometric}. The action of $\Conn{}{}{}$ on $(1,1)$-tensors $\mcA \in \mcX(T^{*}Q \otimes TQ)$ (such as a projection map) is given by
\begin{align}\label{covariant_d_tensor}
    (\Conn{}{\bX}{\mcA}) = \left(\pp{\mcA^i_j}{q^k}X^k + \Gamma^i_{kl}\mcA^l_jX^k - \Gamma^l_{kj}\mcA^i_lX^k\right)\bpp{}{q^{i}}\otimes \bdd q^{i}.
\end{align}  We may consider the connection coefficients as the point-wise bilinear map $\Gamma:T_{\bq}Q\times T_{\bq}Q \rightarrow T_{\bq}Q$ as the function $(\bv_{\bq},\bw_{\bq})\mapsto\Gamma(\bq,\bv_{\bq})\bw_{\bq}$. 

The Lagrangian $L:TQ \rightarrow \mbR$ for a simple mechanical system is defined as 
\begin{align}
    L = \frac{1}{2}\mbG_{\bq}(\bv_{\bq},\bv_{\bq}) - V(\bq).
\end{align} The Euler-Lagrange equations of motion for Lagrangian $L$ are equivalent to forced affine connection system $(\nabla,\mbG,V,F)$ 
\begin{align} \label{EL_affine_dynamics}
    \Conn{}{\bv(t)}{\bv(t)} = \left(-\bdd V + F\right)^{\sharp_\mbG}, \quad \dd{}{t}\bq(t) = \bv(t).
\end{align} defined by the connection and the gradient of the potential and external force $F \in T^{*}Q$ \cite{bullo2019geometric}.

\subsubsection{Nonholonomic Affine Connection System}

We review the basic results of nonholonomic affine connection systems treated. For general references, see \cite{bullo2019geometric,lewis1998affine,lewis2000simple}.

Let $\{a^{l}\subset T^{*}Q\ | \ l=1,\dots, m\}$ be a collection of everywhere linearly independent one-forms over $Q$ such that they specify an $(n-m)$-dimensional nonholonomic distribution $\mcD$ at each point $\bq \in Q$ by 
\begin{align}
\label{nonholonomic_constraints}
    \mcD_{\bq} = \left\{\bv_{\bq} \in T_{\bq}Q \ | \ \aaa{a^{l}_{\bq}|\bv_{\bq}} = 0,\quad l =1,\dots,m \right\}.
\end{align} The constraint one-forms are locally given by $a^{l}_{\bq} = \sum_{i= 1}^{n}\mbA^{l}_{i}(\bq)\bdd q^{i}$ for each $l = 1,\dots,m$. A basis for $\mcD_{\bq}$ is chosen by selecting $m$ everywhere linearly independent vector fields, $\{\mbS_{1}(\bq),\dots,\mbS_{n-m}(\bq)\}$ satisfying the constraint equation \eqref{nonholonomic_constraints}. Using the Riemannian metric we complete the basis for $T_{\bq}Q$ by defining 
\begin{align}
    \mcD^{\bot}_{\bq} := \sspan\{(a^1_{\bq})^{\sharp_\mbG},\dots, (a^{p}_{\bq})^{\sharp_\mbG}\}.
\end{align} 
Therefore, each tangent space can be $\mbG$-orthogonally decomposed into $T_{\bq}Q \cong \mcD_{\bq}\oplus \mcD^{\bot}_{\bq}$. Define the projection maps $\mcP_{\bq}:T_{\bq}Q \rightarrow \mcD_{\bq}$ and $\mcP^{\bot}_{\bq}:T_{\bq}Q \rightarrow \mcD^{\bot}_{\bq}$ at each point $\bq \in Q$. The nonholonomic distribution $\mcD$ can be described in terms of the projection maps as
\begin{align}
    \mcD &= \{\bw \in TQ \, | \, \bw = \mcP(\bv), \, \bv \in TQ\} = \{\bw \in TQ \, | \, \mcP^{\bot}(\bw) = \bzero \}.
\end{align}
 To represent the evolution of the trajectory of the system $\bq:[a,b]\rightarrow Q$, we form the nonholonomic affine connection system:
\begin{align}
\label{nonholonomic_affine:eq1}
    & \conn{}_{\bv(t)}\bv(t) = \blamb(t) -\left(\bdd V \right)^{\sharp_\mbG}, \\
    & \bv \in \mcD \iff \mcP^\bot(\bv(t)) = \bzero, \label{nonholonomic_affine:eq2}\\
    & \dd{}{t}\bq(t) = \bv(t), \label{nonholonomic_affine:eq3}
\end{align} where $\blamb$ is a section in $\mcD^\bot$ representing the Lagrange multipliers representing the constraint reaction forces given by
\begin{align}
\label{lagrange_mult_nonholonomic}
    & \blamb(t) = - \left(\conn{}_{\bv}\mcP^\bot\right)(\bv) + \mcP^{\bot}\left(\bdd V \right)^{\sharp_\mbG}. 
\end{align} 

\begin{remark}
    The nonholonomic affine connection \cite{bullo2019geometric,lewis2000simple,lewis2020nonholonomic} is obtained by substituting the Lagrange multipliers \eqref{lagrange_mult_nonholonomic} into the \eqref{nonholonomic_affine:eq1} and restricting the equations of motion \eqref{nonholonomic_affine:eq1}-\eqref{nonholonomic_affine:eq3} to the distribution $\mcD$.
\end{remark}

\subsection{Geometry of the Second-Tangent Bundle}
In this section, we review the geometry of the second tangent bundle $TTQ$. First, we review the bundle structure and decomposition \cite{abraham2008foundations,ryan2014geometry} which is required to compute the covariant derivative of a general nonlinear vector bundle map. The covariant derivative of a general nonlinear vector bundle map is a prerequisite for the main results of this article. It is noteworthy that this specific result isn't explicitly stated in the literature. Lastly, we review the second tangent bundle in the context of geometric mechanics \cite{abraham2008foundations}. For general references, see \cite{abraham2008foundations,bullo2019geometric,frankel2011geometry,marsden1998introduction,yano1973differential}. 

\subsubsection{Vertical/Horizontal Decomposition of $TTQ$}
Consider the \emph{second tangent bundle} of $Q$ with bundle structure given by $\pi_{TTQ}:TTQ \rightarrow TQ$. Local coordinates for $T_{(\bq,\bv_{\bq})}TQ$ by $\left(\bq,\bv_{\bq},\bV_{(\bq,\bv_{\bq})},\bW_{(\bq,\bv_{\bq})} \right)$ where $(\bV_{(\bq,\bv_{\bq})},\bW_{(\bq,\bv_{\bq})})$ given in the canonical coordinate induced basis by
\begin{align}
    \small\bV_{(\bq,\bv_{\bq})} = \sum_{i=1}^{n}V^{i}\pp{}{\tl{q}^i}\Big |_{(\bq,\bv_{\bq})}, \quad \bW_{(\bq,\bv_{\bq})} = \sum_{j=1}^{n}W^{j}\pp{}{v^{j}}\Big |_{(\bq,\bv_{\bq})}.
\end{align} where $\tl{q}^{i} :=\pi^{*}_{TQ}q^{i}$. The push-forward of the projection map $\pi_{TQ}$ induces a bundle structure $\left(\pi_{TQ}\right)_{*}:TTQ \rightarrow TQ$. In these coordinates, the bundle projection maps satisfy 
\begin{align}
    &\pi_{TTQ}: \left(\bq,\bv,\bV,\bW \right) \mapsto \left(\bq,\bv \right)\\
    &\left(\pi_{TQ} \right)_{*}: \left(\bq,\bv,\bV,\bW  \right) \mapsto \left(\bq,\bV \right).
\end{align} The vertical bundle is defined by $\mcV TQ = \ker\left(\pi_{TQ} \right)_{*}$ and given locally by 
\begin{align} \notag
    \small\mcV TQ &= \left\{\left(\bq,\bv,\bV,\bW \right) \in TTQ \ | \ \left(\pi_{TQ} \right)_{*} \left(\bq,\bv,\bV,\bW \right) = \bzero \right\}\\
    &= \sspan\left\{\bpp{}{v^i}\Big |_{(\bq,\bv_{\bq})} \, | \, i = 1,\dots,n \right\}.
\end{align} 
The parallel transport equation induces a vector-valued connection one form $\omega\in \Omega^{1}\left(TQ;\text{End}(\mbR^{n})\right)$ given by (see section 9.7 of \cite{frankel2011geometry})
\begin{align}
\label{Hor_condition}
    \omega := \bdd v^k + \Gamma^k_{ij}v^j\bdd x^i.
\end{align} Define the \emph{horizontal distribution} $\mcH TQ = \ker(\omega) \subset TTQ$  \cite{frankel2011geometry} which is given locally by 
\begin{align}
    &\mcH TQ = \underset{i =1,\dots,n}{\sspan}\left\{\bpp{}{\tl{q}^{i}}\Big |_{(\bq,\bv_{\bq})} - \Gamma^{k}_{ij}v^{j}\bpp{}{v^{k}}\Big |_{(\bq,\bv_{\bq})} \right\}.
\end{align} Therefore, the second tangent bundle is decomposed as $TTQ \cong \mcH TQ\cong \mcV TQ$. Define the vertical lift \cite{lewis2000simple} of a vector field $\bX\in TQ$ to the vertical space $\mcV_{(\bq,\bv_{\bq})}TQ$
\begin{align} \label{vert_lift}\notag
    &\lift{\ver}:\ TQ \rightarrow \mcV_{(\bq,\bv_{\bq})}TQ\\
    & \bX \mapsto \dd{}{t}\left(\bv_{\bq} + t\bX \right)\Big |_{t=0} .
\end{align} In coordinates, the vertical lift is given by $\bX^{\lift{\ver}} = X^i\bpp{}{v^i}$. The vertical lift is a smooth bijection and hence provides an isomorphism between the $\mcV_{(\bq,\bv_{\bq})}TQ$ and $TQ$ given by $\down{\ver}:= \left(\lift{\ver} \right)^{-1}$
\begin{align} \label{vert_drop} \notag
    &\down{\ver}:\mcV_{(\bq,\bv_{\bq})}TQ \rightarrow  TQ\\
    &\left(\bq,\bv_{\bq},\bzero,\bW \right) \mapsto \left(\bq,\bW \right). 
\end{align} Similarly, we define the \emph{horizontal lift} \cite{yano1973differential} with respect to the affine connection $\Conn{}{}{}$ with Christoffel Symbols $\Gamma$ by
\begin{align} \label{hor_lift} \notag
    &\lift{\hor}: T_{q}Q\times TQ \rightarrow \mcH_{(\bq,\bv_{\bq})}TQ\\
    &\left(\bq,\bv_{\bq},\bX \right) \mapsto \left(\bq,\bv_{\bq},\bX,- \Gamma(\bq,\bv_{\bq})\bX \right).
\end{align} and its inverse $\down{\hor} := \left(\lift{\hor} \right)^{-1}$ given by the map
\begin{align} \label{hor_drop} \notag
    &\down{\hor}: \ \mcH_{(\bq,\bv_{\bq})}TQ \rightarrow TQ\\
    &\left(\bq,\bv_{\bq},\bX, - \Gamma(\bq,\bv_{\bq})\bX\right) \mapsto \left(\bq,\bX \right).
\end{align} Lastly, we define projection maps from the tangent of the tangent bundle onto its vertical and horizontal subbundles. The projection onto the vertical bundle is defined intrinsically by $\mcQ^{\mcV}_{\Gamma} = id_{\mcV TQ}\oplus \Gamma$ and given locally by 
\begin{align} \label{vert_proj} \notag
    &\mcQ^{\mcV}_{\Gamma}: TTQ \rightarrow \mcV TQ\\
    &\left(\bq,\bv_{\bq},\bV,\bW \right) \mapsto \left(\bq,\bv_{\bq},\bzero, \bW + \Gamma(\bq,\bv_{\bq})\bV  \right).
\end{align} From the projection property, the projection onto the horizontal space $\mcH TQ $ is defined intrinsically as $\mcQ_{\Gamma}^{\mcH}:= id_{TTQ} - \mcQ_{\Gamma}^{\mcV}$ and given locally by 
\begin{align} \label{hor_proj} \notag
    &\mcQ_{\Gamma}^{\mcH}: TTQ\rightarrow \mcH_{(\bq,\bv_{\bq})}TQ\\
    & \left(\bq,\bv_{\bq},\bV,\bW \right) \rightarrow \left(\bq,\bv_{\bq},\bV,-\Gamma(\bq,\bv_{\bq})\bV \right).
\end{align} From these projection maps, we define projection from $TTQ$ to two copies of $TQ$, each of which is isomorphic to $\mcV TQ$ and $\mcH TQ$, respectively. These projections are defined in the following diagram:
\begin{center}
\[\begin{tikzcd}
	TTQ && {\mathcal{V}TQ} && TQ \\
	TTQ && {\mathcal{H} TQ} && TQ
	\arrow["{\mathcal{Q}^{\mathcal{V}}_{\Gamma}}"', from=1-1, to=1-3]
	\arrow["{\downarrow_{\text{ver}}}"', from=1-3, to=1-5]
	\arrow["{\mcP_{\mathcal{V}}}", curve={height=-18pt}, from=1-1, to=1-5]
	\arrow["{\mathcal{Q}^{\mathcal{H}}_{\Gamma}}", from=2-1, to=2-3]
	\arrow["{\downarrow_{\text{hor}}}", from=2-3, to=2-5]
	\arrow["{\mcP_{\mathcal{H}}}"', curve={height=18pt}, from=2-1, to=2-5]
\end{tikzcd}\]
\end{center} Therefore, the isomorphism $TTQ \cong \pi_{TQ}^{*}TQ\oplus \pi_{TQ}^{*}TQ$ is established by $\mcP_{\mcH}\oplus\mcP_{\mcV}$. 
\begin{prop}\label{prop_projection}
    The projection maps $\mcP_{\mcH}:TTQ \rightarrow TQ$ and $\mcP_{\mcV}:TTQ \rightarrow TQ$ satisfy the following properties: 
    \begin{itemize}
        \item $\mcP_{\mcH}  = \left(\pi_{TQ}\right)_{*}$
        \item $\mcP_{\mcV}(\bX_{*}) = \Conn{}{\bv}{\bX}$, for a section $\bX:Q \rightarrow TQ$.
    \end{itemize}
\end{prop}
\begin{proof}
    See appendix \ref{proof_prop_projection}
\end{proof}
\subsubsection{Covariant Derivative of a Tangent Bundle Map}
In this subsection, we find the covariant derivative of a general (nonlinear) tangent bundle map. This covariant derivative is essential for the results in subsequent sections. The authors have not found this result explicitly stated in the literature. Lastly, we provide a matrix formulation of the equations for practical computations.  

\begin{prop} \label{prop_bundle_map}
Let $\bh:TQ \rightarrow TQ$ be a general (not necessarily linear) vector bundle map over the tangent bundle $\pi_{TQ}:TQ \rightarrow Q$. Let $\bv,\,\bw:Q \rightarrow TQ$ be a sections. Then, the covariant derivative of $\bh\circ\bw$ along $\bv$ is     \begin{align}
    \Conn{}{\bv}{(\bh\circ \bw)} =  \nabla^{\mcH}_{\bv}\bh(\bw) + D^{\mcV}\bh\left(\Conn{}{\bv}{\bw} \right), 
\end{align} where the vertical Jacobian $D^{\mcV}\bh = \mcP_{\mcV}D\bh\big |_{\mcV TQ}$ is given locally by 
\begin{align} \label{vertical_deriv}
    D^{\mcV}\bh(\Conn{}{\bv}{\bw}) = \pp{\bh^m}{w^l}\left(\Conn{}{\bv}{\bw} \right)^l \bpp{}{q^m},
\end{align} and the horizontal connection $\Conn{\mcH}{\bv}\bh(\bw) = \mcP_{\mcV}D\bh\big|_{\mcH TQ}$ is given locally by 
\begin{align}\label{horizontal_deriv}
    \Conn{\mcH}{\bv}{\bh(\bw)} = \left(\pp{\bh^m}{q^{k}}v^k  + \Gamma^{m}_{ki}v^{k}\bh^{i} - \Gamma^{l}_{ki}w^{k}v^{i}\pp{\bh^m}{w^l}\right)\bpp{}{q^m}.
\end{align}
\end{prop}
\begin{proof}
    See appendix \ref{proof_prop_bundle_map}
\end{proof}

\begin{remark}
    Note that when the bundle map $\bh:TQ \rightarrow TQ$ linear in velocities, the vertical derivative $D^{\mcV}\bh = \bh$ is linear in its argument, and the horizontal covariant derivative $\Conn{\mcH}{\bX}{\bh}$ reduces to the usual covariant derivative $\Conn{}{\bX}{\bh}$ e.g., see equation \eqref{covariant_d_tensor}. Therefore, this "chain rule" is a generalization of the traditional chain rule for the covariant derivative of a linear vector bundle map. 
\end{remark}

\begin{remark}\label{covariant_remark}
    For practical calculations, in the coordinate induced basis $\{\bpp{}{q^{1}},\dots,\bpp{}{q^{n}}\}$, the covariant derivative $\Conn{}{\bX}{\bY}$ given in \eqref{covariant_d_vector} may be calculated in matrix form as
    \begin{align}\label{conn_vec_matrix}
        \Conn{}{\bX}{\bY} = \left[\pp{\bY}{\bq}\right]\bX + \left[\bGamma(\bq,\bX)\right]\bY, 
    \end{align} where $[\pp{\bY}{\bq}] = [\pp{Y^i}{q^k}]$ is the Jacobian of $\bX$ and the Christoffel symbols are organized in a matrix $[\bGamma(\bq,\bX)] = [\Gamma^{i}_{kj}X^{j}]$. Similarly, the covariant derivative of a (1,1)-tensor $\mcA$ given in \eqref{covariant_d_tensor} may be calculated in matrix form as
    \begin{align}\label{conn_tensor_matrix}
        [\Conn{}{\bX}{\mcA}] = \left[\sum^{n}_{k=1}\ \pp{\mcA}{q^{k}}X^{k}\right] + [\bGamma(\bq,\bX)][\mcA] - [\mcA][\bGamma(\bq,\bX)].
    \end{align} Lastly, the covariant derivative $\Conn{\mcH}{\bX}{\bh}$ of a tangent bundle map $\bh(\bq,\bw)$ as in \eqref{horizontal_deriv} may be calculated in matrix form by 
    \begin{align}\label{conn_bundle_matrix}
        \Conn{\mcH}{\bX}{\bh} = \left[\pp{\bh}{\bq}(\bq,\bw)\right]\bX + [\bGamma(\bq,\bX)]\bh - \left[D^{\mcV}\bh\right][\bGamma(\bq,\bX)]\bw.
    \end{align}
\end{remark}

\subsubsection{Dynamics of Mechanical Systems on $TTQ$}
A vector field $\bX:TQ \rightarrow TTQ$ is a \emph{second order system} on $Q$ if and only if for each integral curve $\bq:[a,b]\rightarrow TQ$ of $\bX$ we have $\dd{}{t}\left(\pi_{TQ}\circ \bq \right)(t) =  \left(\pi_{TQ} \right)_{*}\dot{\bq} = \bq(t)$, summarized in the following diagram \cite{abraham2008foundations}
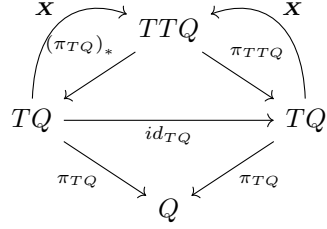
\begin{center}
\[\begin{tikzcd}
	& TTQ \\
	TQ && TQ && {} \\
	& Q
	\arrow["{\left(\pi_{TQ}\right)_{*}}"'{pos=0.2}, from=1-2, to=2-1]
	\arrow["{\pi_{TTQ}}"{pos=0.3}, from=1-2, to=2-3]
	\arrow["{\pi_{TQ}}", from=2-3, to=3-2]
	\arrow["{\pi_{TQ}}"', from=2-1, to=3-2]
	\arrow["{\boldsymbol{X}}"', curve={height=30pt}, from=2-3, to=1-2]
	\arrow["{\boldsymbol{X}}", curve={height=-30pt}, from=2-1, to=1-2]
	\arrow["{id_{TQ}}"', from=2-1, to=2-3]
\end{tikzcd}\]
\end{center} 
If $\bX = \left(\bq,\bv,\bV,\bW \right)$ is a second-order vector field, then it follows that the velocity components satisfy $v^{i} = V^{i}$. 

Given, the decomposition $TTQ \cong \mcH TQ \oplus \mcV TQ$, we construct the vector field associated with the affine connection system \eqref{EL_affine_dynamics}. Let $\left(\bq,\bv\right):[a,b] \rightarrow TQ$ be a curve of satisfying \eqref{EL_affine_dynamics}. The forces acting on the system $\left(-\bdd V^{\sharp_{\mbG}} +F^{\sharp_{\mbG}} \right)$ in the right-hand side of \eqref{EL_affine_dynamics} are vertically lifted to $\mcV_{(\bq,\bv)}TQ$ by \eqref{vert_lift}. The vector field associated with the parallel transport equation $\Conn{}{\bv(t)}{\bv(t)} = \bzero $ on the left-hand side of \eqref{EL_affine_dynamics} is given by the horizontal lift \eqref{hor_lift} of the velocity $\bv$ to $\mcH_{(\bq,\bv_{\bq})}TQ$.
Therefore, the time variation of the curve $\dd{}{t}\left(\bq,\bv\right) \in T_{(\bq,\bv)}TQ \cong \mcH_{(\bq,\bv)}TQ\oplus\mcV_{(\bq,\bv)}TQ$ satisfies the differential equations 

\begin{align}
\label{second_order_vf}
    & \dd{\bv}{t} = -\Gamma^{k}_{ij}v^{i}v^{j}\bpp{}{v^{k}} + \left( - \pp{V}{q^{i}} + F_{i} \right)\mbG^{ik} \bpp{}{v^{k}},\\
    &\dd{\bq}{t} = v^{i}\bpp{}{\tl{q}^{i}}.
\end{align} 
\begin{remark}
    When there are no external forces in \eqref{second_order_vf}, the vector field $\bZ:TQ \rightarrow TTQ$
    \begin{align}
        \bZ(\bq,\bv) = v^{i}\bpp{}{\tl{q}^{i}} - \Gamma^{k}_{ij}v^{i}v^{j}\bpp{}{v^{k}},
    \end{align} is called the \emph{geodesic spray} \cite{bullo2019geometric,lewis2000simple}. The integral curves of the geodesic spray $\bZ$ in $TQ$ project to geodesics with respect to $\Conn{}{}{}$ in $Q$.
\end{remark}



\section{Affine Connection System as a Slow/Fast System with Strong Friction} \label{sec_affine_connection_systems} 

Let $Q$ be a smooth manifold and let $\mcD\subset TQ$ be a linear subbundle representing a set of $m$ ideal nonholonomic constraints. To realize the nonholonomic constraints by friction forces, consider a Rayleigh dissipation function $\mcR:T_qQ\times T_qQ \rightarrow \mbR$, represented as a symmetric (0,2)-tensor on each tangent space \cite{eldering2016realizing}. This function induces a dissipative friction force given by the (negative) fibre derivative $-\mbF\mcR:TQ \rightarrow T^*Q$ of the dissipation function. Assume that this friction force satisfies $\mcD = \ker\left(\mbF\mcR\right)$ to ensure that the nonholonomic distribution is attractive. 

We construct this dissipation function $\mcR$  as follows. Consider a $\mbR^{m}$-valued one-form $a_{\bq}:T_{\bq}Q \rightarrow \mbR^{m}$ consisting of the $m$ nonholonomic constraint one-forms $a^{l} \in T^*Q$ for $l = 1,\dots,m$. For each configuration $\bq \in Q$, define the friction coefficients as the symmetric linear vector bundle map $\bar{\mu}(\bq): \mbR^{m} \rightarrow \mbR^{m*} \cong \mbR^{m}$ and finally, define the adjoint of the map representing the constraint as $a^{*}_{\bq}:\mbR^{m*}\rightarrow T^{*}_{\bq}Q$. Define the Rayleigh dissipation function $\mcR = \frac{1}{2}a^{*}_{\bq}\circ \mu(q) \circ a_{\bq}$ and in local coordinates is given by 
\begin{align} \label{Rayleigh_disp}
    \mcR = \frac{1}{2}\underbrace{\mbA_{i}^{l}(\bq)\bar{\mu}_{lk}(\bq)\mbA^{k}_{j}(\bq)}_{\mcR_{ij}(\bq)}\bdd q^{j}\otimes\bdd q^{i}.
\end{align}
The dissipative Rayleigh friction force with large coefficients is defined as follows. Fix a small positive real number $\eps >0$ (which is thought to be the time scale on which this force acts) and rescale the friction coefficients $\bar{\mu} = 1/\eps\mu$ to find the friction force given by 
\begin{align}  \label{Rayleigh_force}
    F_\eps = -\frac{1}{\eps}\mbF\mcR(\bv) = - \frac{1}{\eps}\mcR_{ij}(\bq)v^{i}\bdd q^{j} \in T^*Q.
\end{align} Let $\bq:[a,b] \rightarrow Q$ be a smooth curve that satisfies the equations of motion formulated by the forced affine connection system 
\begin{align}
\label{fric_con_sys:eq1}
&\Conn{}{\bv}{\bv} = -\left(\bdd V \right)^{{\sharp_{\mbG}}}  - \frac{1}{\eps}\mbF\mcR(\bv)^{\sharp_{\mbG}},\\
& \dd{}{t}\bq = \bv.\label{fric_con_sys:eq2}
\end{align}

 As this is a singularly perturbed system, it follows that (\ref{fric_con_sys:eq1}-\ref{fric_con_sys:eq2} emits a \emph{normally hyperbolic invariant manifold} (NHIM) for $\eps >0$ \cite{eldering2016realizing}. That is, for small values of $\eps$, there exists a perturbed manifold $\mcD_\eps$ that is $C^1$-close and diffeomorphic to the nonholonomic constraint distribution $\mcD$. The existence of this submanifold has been studied by Eldering and summarize the result \cite{eldering2016realizing} as follows. Write the system as a vector field \eqref{second_order_vf} on $TQ$ according to the decomposition $TTQ \cong \mcH TQ \oplus \mcV TQ$: 
 \begin{align} \label{geodesic_spray}
     \small\dd{}{t}\left(\bq,\bv\right) = \bZ(\bq,\bv) + \left(-\bdd V^{{\sharp_{\mbG}}}\right)^{\lift{\ver}} - \left(\frac{1}{\eps}\mbF\mcR(\bv)^{\sharp_{\mbG}}\right)^{\lift{\ver}}
 \end{align} where we have vertically lifted the forces to $TTQ$ using \eqref{vert_lift}. Make the reparameterization of the trajectory to fast-time $\tau = t/\eps$ to find  
 \begin{align}
     \small\dd{}{\tau}\left(\bq,\bv\right) = \eps \bZ(\bq,\bv) +\eps \left(-\bdd V^{{\sharp_{\mbG}}}+ U^{\sharp_{\mbG}}\right)^{\lift{\ver}} -  \left(\mbF\mcR(\bv)^{\sharp_\mbG}\right)^{\lift{\ver}}.
 \end{align} When $\eps = 0$ then the fast time system is of the form 
 \begin{align}\label{fast_time_dynamics}
     \dd{}{\tau}\left(\bq,\bv\right) = - \left(\mbF\mcR(\bv)^{\sharp_\mbG}\right)^{\lift{\ver}}.
 \end{align} When this friction force acts on an instantaneous time-scale ($\eps = 0$) it is the dominant force acting on the system. By construction, when we restrict the dynamics to the distribution $\mcD$, the friction force vanishes; thus $\mcD$ is an invariant manifold for the fast-time dynamics (\ref{fast_time_dynamics}). When the dynamics are restricted to the $\mbG$-orthogonal complement $\mcD^{\bot}$, we find that the dynamics are attractive. Therefore, by the theory of \emph{normally hyperbolic invariant manifolds} \cite{eldering2013normally} there exists a perturbed manifold $\mcD_\eps$ that is $C^1$-close and diffeomorphic to $\mcD$.

 \subsection{Representation of the Slow Manifold}
 We represent the perturbed manifold $\mcD_{\eps}$ as the image of a section $\tl{\bH}_{\eps}:\mcD \rightarrow TQ$ over the bundle $\mcP:TQ \rightarrow \mcD$. By definition, the section $\tl{\bH}_{\eps}$ satisfies the following commutative diagram
\begin{center}
\[\begin{tikzcd}
	{\mathcal{D}} && TQ \\
	& {\mathcal{D}}
	\arrow["{\tilde{\boldsymbol{H}}_{\varepsilon}}", from=1-1, to=1-3]
	\arrow["{\mathcal{P}}", from=1-3, to=2-2]
	\arrow["{\text{id}_{\mathcal{D}}}"', from=1-1, to=2-2]
\end{tikzcd}\]
\end{center} That is, $\mcP\circ\tl{\bH}_{\eps} = \text{id}_{\mcD}$. By the projection property $\text{id}_{TQ} = \mcP\oplus \mcP^{\bot}$, we find that the section $\tl{\bH}_{\eps}$ is given by 
 \begin{align}
     \tl{\bH}_{\eps} = \text{id}_{\mcD} + \underbrace{\mcP^{\bot}\circ\tl{\bH}_{\eps}}_{=:\bh_{\eps}}
 \end{align} where $\bh_{\eps}:\mcD \rightarrow\mcD^{\bot}$ is a section over $\pi_{TQ}:TQ \rightarrow Q$. Therefore, the slow manifold $\mcD_{\eps}$ is defined by the set 
 \begin{align}
     \mcD_{\eps} = \left\{\bv \in TQ \, | \, \bv = \bv^{\mcD} + \bh_{\eps}(\bv^{\mcD}), \, \bv^{\mcD} \in \mcD \right\}.
 \end{align}

 Using the projection maps, we then define the sections $\bH_{\eps} = \tl{\bH}_{\eps}\circ\mcP:TQ\rightarrow TQ$ as 
 \begin{align}
     \bH_{\eps}(\bv) = \mcP(\bv) + \bh_{\eps}\circ\mcP(\bv),
 \end{align} where $\bh_{\eps}\circ\mcP = \mcP^{\bot}\circ\tl{\bH}_{\eps}\circ\mcP:TQ \rightarrow \mcD^{\bot}$. Therefore, the perturbed manifold $\mcD_{\eps} = \mbox{im}(\bH_{\eps})$ is given by 
 \begin{align}
\mcD_{\eps} = \left\{\bw \in TQ\, | \, \bw = \bH_{\eps}(\bv) = \mcP(\bv) + \bh_{\eps}\circ\mcP(\bv), \, \bv \in TQ \right\}.
 \end{align} Note that the form of this function $\bh_{\eps}$ depends on the choice of projection map $\mcP$. 
 \begin{remark}
     Physically speaking, the slow manifold represents the slip velocities in the orthogonal complement $\mcD^{\bot}$ and are a function of the variables in the nonholonomic distribution $\mcD$. 
 \end{remark}

\begin{lemma}\label{prop_nonlin_proj}
    The map $\bH_{\eps}:TQ \rightarrow TQ$ for the system \eqref{fric_con_sys:eq1}-\eqref{fric_con_sys:eq2} is a \emph{nonlinear projection} onto $\mcD_{\eps}$. 
\end{lemma}
\begin{proof}
    Given a section $\bv:Q \rightarrow TQ$ we find
     \begin{align} \notag
         \bH_{\eps}\circ\bH_{\eps}(\bv) &= \bH_{\eps}\left(\mcP(\bv) + \bh_{\eps}\circ\mcP(\bv) \right)\\  \notag
         &= (\tl{\bH}_{\eps}\circ\mcP)\left(\mcP(\bv) + \bh_{\eps}\circ\mcP(\bv) \right)\\  \notag
         &= \tl{\bH}_{\eps}\left(\mcP\circ\mcP(\bv) + \mcP\circ\bh_{\eps}\circ\mcP(\bv) \right)\\ \notag
         &= \tl{\bH}_{\eps}\left(\mcP\circ\mcP(\bv) + \mcP\circ \mcP^{\bot}\circ\tilde{\bH}_{\eps}\circ \mcP(\bv) \right)\\ \notag
         & = \tl{\bH}_{\eps}\circ\mcP(\bv) = \bH_{\eps}(\bv). \notag
     \end{align}
\end{proof}  
\begin{remark}
    Since $\bH_{\eps}$ defines a nonlinear projection map, we may also define its complement $\bH^{\bot}_{\eps} := id_{Q} - \bH_{\eps}$. This defines a projection map onto the complement of the slow manifold abbreviated by $TQ/\mcD_{\eps}$. 
\end{remark}
\begin{lemma}\label{prop_nonlin_proj_comp} The map $\bH^{\bot}_{\eps}:TQ\rightarrow TQ/\mcD_{\eps}$ is a nonlinear projection map that is complementary to $\bH_{\eps}$.
\end{lemma}
\begin{proof}
    First, we show that $\bH_{\eps}$ is a projection map. Let $\bv:Q\rightarrow TQ$ be a section. Then, 
    \begin{align} \notag
        \bH_{\eps}^{\bot}\circ\bH_{\eps}^{\bot}(\bv)&=\left(id_{TQ} - \bH_{\eps} \right)\circ(\bv - \bH_{\eps}(\bv))\\ \notag
        &=\bv - \bH_{\eps}(\bv) - \bH_{\eps}(\bv) + \bH_{\eps}\circ \bH_{\eps}(\bv)\\ \notag
        &= \bv - \bH_{\eps}(\bv) = \bH_{\eps}^{\bot}(\bv). \notag
    \end{align} Lastly, we show that this map is complement to $\bH_{\eps}$.
    \begin{align}\notag
        &\bH_{\eps}\circ\bH_{\eps}^{\bot}(\bv) = \bH_{\eps}(\bv) - \bH_{\eps}\circ\bH_{\eps}(\bv) = 0,\\ \notag
        &\bH_{\eps}^{\bot}\circ\bH_{\eps}(\bv) = \bH_{\eps}\circ\bH_{\eps}(\bv) - \bH_{\eps}(\bv)   = 0.\notag
    \end{align}
\end{proof}
\begin{remark}
    As a consequence of the complementary projection map in Proposition \ref{prop_nonlin_proj_comp}, we find that $\bv \in \mcD_{\eps}$ if and only if 
    \begin{align} \notag
        \bH_{\eps}^{\bot}(\bv) &= \bv -\bH_{\eps}(\bv) = \bv - \mcP(\bv) - \bh_{\eps}\circ\mcP(\bv)\\ \notag
        &= \mcP^{\bot}(\bv) - \bh_{\eps}\circ\mcP(\bv) = 0.
    \end{align} This is summarized in the following proposition.  
\end{remark}

 \begin{lemma} \label{prop:perturbed_dist}
     $\bv \in \mcD_{\eps}$ if and only if $\mcP^{\bot}(\bv) = \bh_{\eps}\circ\mcP(\bv)$.
 \end{lemma}
 \begin{proof} Suppose that $\bw \in \mcD_{\eps}$. Then, there exists $\bv \in TQ$ such that $\bw = \mcP(\bv) + \bh_{\eps}(\bv)$. We calculate
     \begin{align} \notag
     \small\mcP^{\bot}(\bw) &= \mcP^{\bot}\left(\mcP(\bv) + \mcP^{\bot}\circ\bh_{\eps}(\bv)\right)\\ \notag
     &= \mcP^{\bot}\circ\mcP^{\bot}\circ\bH_{\eps}(\bv) = \bh_{\eps}(\bv) = \mcP^{\bot}\circ \tl{\bH}_{\eps}\circ \mcP(\bv).
 \end{align} Moreover, we find that $\bw$ and $\bv$ project to the same vector
 \begin{align}
     \mcP(\bw) = \mcP\left(\mcP(\bv) + \bh_{\eps}\circ\mcP(\bv) \right)) = \mcP(\bv).
 \end{align} Therefore, $\mcP^{\bot}(\bw) = \mcP^{\bot}\circ \tl{\bH}_{\eps}\circ\mcP(\bv) = \bh_{\eps}\circ\mcP(\bw)$.
 
 Conversely, suppose $\bv \in TQ$ such that $\mcP^{\bot}(\bv) = \bh_{\eps}\circ\mcP(\bv)$. Then, since $\mcP$ is a projection map, we have that $\mbid_{TQ} = \mcP\oplus \mcP^{\bot}$. Therefore, since we assume that $\mcP^{\bot}(\bv) = \bh_{\eps}\circ\mcP(\bv)$, we find
 \begin{align}
     \bv = \mcP(\bv) + \mcP^{\bot}(\bv) = \mcP(\bv) + \bh_{\eps}\circ\mcP(\bv). \notag
 \end{align} Therefore, $\bv \in \mcD_{\eps}$.
 \end{proof} By Proposition \ref{prop:perturbed_dist}, we find that the distribution can be represented as the set
\begin{align}
    \mcD_{\eps} = \left\{\bv \in TQ \, | \, \mcP^{\bot}(\bv) = \bh_{\eps}\circ\mcP(\bv) \right\}.
\end{align}
 
\begin{prop}
    The equations of motion for a curve $\bq:[a,b] \rightarrow Q$ lying on the slow manifold $\mcD_\eps$ are given by
\begin{align}\label{slow_fast_eom:eq1}
    &\Conn{}{\bv(t)}{\bv(t)} = -\bdd V^{{\sharp_{\mbG}}}  - \frac{1}{\eps}\mbF\mcR\left(\bv(t)\right)^{\sharp_\mbG}\\
    & \dd{}{t}\bq(t) = \bv(t),\label{slow_fast_eom:eq2}\\
    &\mcP^\bot(\bv(t)) = \bh_{\eps}\circ\mcP(\bv(t)) \iff \bv(t) \in \mcD_{\eps}.\label{slow_fast_eom:eq3}
\end{align} 
\end{prop}
\begin{proof}
    The equations of motion \eqref{slow_fast_eom:eq1}-\eqref{slow_fast_eom:eq2} are described by a forced affine connection system. The existence of the slow manifold follows from the geometric singular perturbation theory of normally hyperbolic manifolds \cite{eldering2013normally}. By Proposition \ref{prop:perturbed_dist} the slow manifold is represented by \eqref{slow_fast_eom:eq3}.
\end{proof}




\section{Approximation of the Slow Subsystem} \label{sec_approximations_slip}

     In this section, we find a generating equation for the invariant manifold $\mcD_\eps$ using the equations of motion \eqref{slow_fast_eom:eq1} - \eqref{slow_fast_eom:eq3}.  The strategy is to consider velocities restricted to the slow manifold $\mcD_\eps$ and expand the section $\bh_\eps$ as a power series and match terms to the desired order in $\epsilon$. 

\subsection{Generating equation for the slow-manifold}
Let us find an expression for the dissipation function $\left(\mbF\mcR\right)^{\sharp_{\mbG}}$ in terms of the unknown function $\bh_\eps$ describing the perturbation of the ideal constraints. Using the projection maps, we project the dynamical equations \eqref{slow_fast_eom:eq1} onto $\mcD\oplus\mcD^{\bot} \cong TQ$
\begin{align}
    &\mcP\left(\Conn{}{\bv}{\bv} \right) = -\mcP\left(\bdd V^{\sharp_{\mbG}} \right)\\
    &\mcP^{\bot}\left(\Conn{}{\bv}{\bv}\right) = -\mcP^{\bot}\left(\bdd V^{\sharp_{\mbG}} \right) - \frac{1}{\eps}\mbF\mcR^{\sharp_{\mbG}}(\bv).
\end{align} 

\begin{lemma}\label{prop_first_isomorphism}
Consider the transformation $\mbF\mcR^{\sharp_{\mbG}}_{\bq}:T_{\bq}Q \rightarrow T_{\bq}Q$, then there exists a transformation $\mbQ_{\bq}:\mcD^{\bot}_{\bq} \rightarrow\mcD^{\bot}_{\bq}$ satisfying $\mbQ_{\bq}\circ \mbF\mcR^{\sharp_{\mbG}}_{\bq} = \mcP^{\bot}_{\bq}$.
\end{lemma} 
\begin{proof}
    The friction force is a linear transformation $\mbF\mcR^{\sharp_{\mbG}}_{\bq}:T_{\bq}Q \rightarrow T_{\bq}Q$ satisfying $\mcD = \ker(\mbF\mcR^{\sharp_{\mbG}}_{\bq})$ and $\text{im}(\mbF\mcR^{\sharp_{\mbG}}_{\bq}) = \mcD^{\bot}$. Thus, by the first isomorphism theorem \cite{dummit2004abstract}, we have the result.
\end{proof} 
\begin{remark} \label{Q_remark}
    The map $\mbQ_{\bq}$ may be constructed as follows. Let $\Phi(\bq)$ be a frame spanning $\mcD_{\bq}\oplus\mcD^{\bot}_{\bq}$. Define linear maps in this frame with the subscript $\Phi$, so, $[\mbQ_{\Phi}] = \Phi^{-1}[\mbQ]\Phi$, $[\mcP^{\bot}_{\Phi}] = \Phi^{-1}[\mcP^{\bot}]\Phi$ and $[\mbF\mcR_{\Phi}^{\sharp_{\mbG}}] = \Phi^{-1}[\mbF\mcR^{\sharp_{\mbG}}]\Phi$ and the equation $[\mcP_{\Phi}^{\bot}]=[\mbQ_{\Phi}][\mbF\mcR_{\Phi}^{\sharp_{\mbG}}]$ is given by the block matrix equation of the form
    \begin{align}
        \underbrace{\begin{bmatrix}
            \mbO & \mbO\\
        \mbO & \mbI
        \end{bmatrix}}_{[\mcP^{\bot}_{\bq}]}  = \underbrace{\begin{bmatrix}
        \mbO & \mbO\\
        \mbO & \tl{\mbQ}
    \end{bmatrix}}_{[\mbQ_{\Phi}]}\underbrace{\begin{bmatrix}
        \mbO & \mbO\\
        \mbO & \widetilde{\mbF\mcR}
    \end{bmatrix}}_{[\mbF\mcR_{\Phi}^{\sharp_{\mbG}}]}.
    \end{align} Therefore, we set $\tl{\mbQ} = \widetilde{\mbF\mcR}^{-1}$ and perform the inverse transformation $[\mbQ_{\bq}] = \Phi[\mbQ_{\Phi}]\Phi^{-1}$ to obtain the result.
\end{remark}

\begin{theorem}[Generating Equation for the Slow Manifold] \label{thm_generating_eq}
    Let $\bh_{\eps}:\mcD \rightarrow TQ$ be the section defining the invariant manifold $\mcD_{\eps}$ with respect to the dynamics \eqref{slow_fast_eom:eq1}-\eqref{slow_fast_eom:eq3}. Let $\bh_{\eps} = \sum_{k \ge 0}\eps^{k}\bh^{(k)}$ be an expansion of the section by a formal power series. Then, the slow manifold may be recursively approximated from the following generating equation 
    \begin{align} \label{generating_eq} \notag
    &\bh_{\eps}(\bv^{\mcD}) = \eps \mbQ_{\bq}\circ\left[(\Conn{}{\bv^{\mcD}}{\mcP^{\bot}})(\bv^{\mcD}) - \mcP^{\bot}(\bdd V^{\sharp_{\mbG}}) \right] \\ \notag
    &\small + \eps\mbQ_{\bq}\circ\left[ (\Conn{}{\bv^{\mcD}}{\mcP^{\bot}})(\bh_{\eps}) + (\Conn{}{\bh_{\eps}}{\mcP^{\bot}})(\bv^{\mcD}) + (\Conn{}{\bh_{\eps}}{\mcP^{\bot}})(\bh_{\eps}) \right] \\ \notag
    &\small- \eps\mbQ_{\bq}\circ D^{\mcV}\bh_{\eps}\left[(\Conn{}{\bv^{\mcD}}{\mcP})(\bv^{\mcD})  - \mcP(\bdd V^{\sharp_{\mbG}}) \right] \\
    &\small - \eps\mbQ_{\bq}\circ D^{\mcV}\bh_{\eps}\left[(\Conn{}{\bh_{\eps}}{\mcP})(\bv^{\mcD})  + (\Conn{}{\bv^{\mcD}}{\mcP})(\bh_{\eps}) + (\Conn{}{\bh_{\eps}}{\mcP})(\bh_{\eps})  \right] \\ \notag
    &\quad \quad \quad - \eps \mbQ_{\bq}\circ\left[\Conn{\mcH}{\bv_{\mcD}}{\bh_{\eps}}(\bv^{\mcD}) + \Conn{\mcH}{\bh_{\eps}}{\bh_{\eps}}(\bv^{\mcD}) \right].
\end{align}
\end{theorem}

\begin{proof}
    By analogy with the ideal nonholonomic case \cite{lewis2000simple}, take the covariant derivative of the latter invariant manifold equation \eqref{slow_fast_eom:eq3} 
\begin{align}
    \Conn{}{\bv}{\left(\mcP^{\bot}(\bv)\right)} = \Conn{}{\bv}{\left(\bh_{\eps}\circ\mcP(\bv)\right)}
\end{align} and substitute the affine connection system \eqref{slow_fast_eom:eq1}-\eqref{slow_fast_eom:eq2}. The left-hand side is given by 
\begin{align} \notag
    \Conn{}{\bv}{\left(\mcP^\bot(\bv) \right)}&= \left(\conn{}_{\bv}\mcP^\bot \right)(\bv) + \mcP^\bot(\conn{}_{\bv}\bv)\\ \notag
    &= \left(\conn{}_{\bv}\mcP^\bot \right)(\bv)  -\frac{1}{\eps}\mbF\mcR(\bv)^{\sharp_{\mbG}} - \mcP^\bot\left(\bdd V^{\sharp_{\mbG}} \right).   \notag
\end{align} On the other hand, according to Proposition \ref{prop_bundle_map} and the chain rule, we find the right-hand side 
\begin{align} \notag
&\Conn{}{\bv}{\left(\bh_\eps\circ\mcP(\bv) \right)}= D^{\mcV}\bh_{\eps}\left(\Conn{}{\bv}{\left(\mcP(\bv)\right)} \right)  + \Conn{\mcH}{\bv}{\bh_{\eps}\left(\mcP(\bv)\right)}\\ \notag
&\small = D^{\mcV}\bh_{\eps}\left(\left(\Conn{}{\bv}{\mcP}\right)(\bv) + \mcP\left(\Conn{}{\bv}{\bv}\right) \right)  + \Conn{\mcH}{\bv}{\bh_{\eps}\left(\mcP(\bv)\right)}\\ \notag
&\small = D^{\mcV}\bh_{\eps}\left((\Conn{}{\bv}{\mcP})(\bv) \right)-D^{\mcV}\bh_{\eps}\left(\mcP\left(\bdd V^{\sharp_{\mbG}}\right)\right)  + \Conn{\mcH}{\bv}{\bh_{\eps}(\mcP(\bv))}. \notag
\end{align} Putting both sides of this calculation together, we find the following identity for the dissipative forces 
\begin{align}
\label{friction_force} 
    &-\frac{1}{\eps}\mbF\mcR^{\sharp_{\mbG}}(\bv) = -(\Conn{}{\bv}{\mcP^{\bot}})(\bv) + \mcP^{\bot}(\bdd V^{\sharp_{\mbG}})\\ \notag
    &+ D^{\mcV}\bh_{\eps}\left((\Conn{}{\bv}{\mcP})(\bv) \right) + \Conn{\mcH}{\bv}{\bh_{\eps}(\mcP(\bv))} - D^{\mcV}\bh_{\eps}\left(\mcP(\bdd V^{\sharp_{\mbG}})\right)
\end{align}  

Lastly, we restrict \eqref{friction_force} to the velocities $\bv \in \mcD_{\eps}$ to obtain the generating equation for $\bh_{\eps}$
\begin{align} \label{invariance_equation} \notag
    \mbF\mcR^{\sharp_{\mbG}}\bh_{\eps}(\bv^{\mcD})&= \eps \left(\conn{}_{\bv^{\mcD} + \bh_{\eps}} \mcP^{\bot}\right)(\bv^{\mcD} + \bh_{\eps}) + \eps\mcP^{\bot}(\bdd V^{\sharp_{\mbG}})\\ \notag
    & - \eps D^{\mcV}\bh_{\eps}(-\bdd V^{\sharp_{\mbG}})  - \eps \left(\conn{\mcH}_{\bv^{\mcD} + \bh_{\eps}}\bh_{\eps} \right)(\bv^{\mcD}) \\ 
    & +\eps D^{\mcV}\bh_{\eps}\left( (\Conn{}{\bv^{\mcD} + \bh_{\eps}}{\mcP})(\bv^{\mcD} + \bh_{\eps}) \right).
\end{align} By Lemma \ref{prop_first_isomorphism} and linearity we obtain the generating equation \eqref{generating_eq}.
\end{proof}
\begin{remark}
    Note that there is an $\eps$ in all terms on the right-hand side of \eqref{generating_eq}. Thus, any power series expansion of $\bh_\eps$ results in an order higher on the right-hand side than the left-hand side. Therefore, the generating equation \eqref{generating_eq} provides a recurrence relation for the approximation of $\bh_{\eps}$.
\end{remark}
 
\begin{theorem}[Equivalence of Reaction Forces] \label{thm_equiv_forces} The generating equation for $\bh_{\eps}(\bv^{\mcD})$ in \eqref{generating_eq} results in the same reaction forces $-\frac{1}{\eps}\mbF\mcR^{\sharp_{\mbG}}\bh_{\eps}(\bv^{\mcD})$ as those reaction forces resulting from the classical invariance condition 
\begin{align}\label{classical_invariance_equation}
    \dd{}{t}\mcP^{\bot}(\bv) \big|_{\mcD_{\eps}} = \dd{}{t}\bh_{\eps}\circ\mcP(\bv) \big|_{\mcD_{\eps}}.
\end{align}
\end{theorem}
\begin{proof}
    Without loss of generality, assume that the potential $V=0$. The equations of motion \eqref{slow_fast_eom:eq1}-\eqref{slow_fast_eom:eq3} can be written as 
    \begin{align}
        &\dot{\bv} = -\bGamma(\bq,\bv)\bv - \frac{1}{\eps}\mbF\mcR^{\sharp_{\mbG}}(\bv)\\
        &\dot{\bq}=\bv \in \mcD_{\eps} \iff \mcP^{\bot}(\bv) = \bh_{\eps}\circ\mcP(\bv).
    \end{align} Taking a Lie derivative of \eqref{slow_fast_eom:eq3} along the along these dynamics we find 
    \begin{align} \notag
        &\left[\sum^{n}_{k=1}\pp{\mcP^{\bot}}{q^{k}}v^{k}\right]\bv - [\mcP^{\bot}][\bGamma(\bq,\bv)]\bv - \frac{1}{\eps}[\mcP^{\bot}][\mbF\mcR^{\sharp_{\mbG}}]\bv = \\ \notag
        &\left[\pp{\bh_{\eps}}{\bq}(\bq,\mcP(\bv))\right]\bv + [D^{\mcV}\bh_{\eps}]\left[\sum^{n}_{k=1}\pp{\mcP}{q^k}v^{k}\right]\bv\\ \notag
        &\quad - [D^{\mcV}\bh_{\eps}][\mcP][\bGamma(\bq,\bv)]\bv - \frac{1}{\eps}[D^{\mcV}\bh_{\eps}][\mcP][\mbF\mcR^{\sharp_{\mbG}}]\bv. \notag
    \end{align} where the quantities in the square brackets are the compact matrix notation introduced in Remark \ref{covariant_remark}. As $\mcD = \ker\mbF\mcR^{\sharp_\mbG}$, the last term vanishes. Let $\bv = \bv^{\mcD_{\eps}} \in \mcD_{\eps}$ be a velocity in the slow manifold, then the reaction forces are restricted to $\mcD_{\eps}$ are given by 
    \begin{align}\label{classical_invariance_forces}
        &\small-\frac{1}{\eps}\mbF\mcR^{\sharp_{\mbG}}(\bv^{\mcD_{\eps}}) = -\left[\sum^{n}_{k=1}\pp{\mcP^{\bot}}{q^k}v^{k}_{\mcD_{\eps}}\right]\bv^{\mcD_{\eps}} +[\mcP^{\bot}][\bGamma(\bq,\bv^{\mcD_{\eps}})]\bv^{\mcD_{\eps}}\\ \notag
        & \quad \quad \small +\left[\pp{\bh_{\eps}}{\bq}(\bq,\bv^{\mcD})\right]\bv^{\mcD_{\eps}} + [D^{\mcV}\bh_{\eps}]\left[\sum^{n}_{k=1}\pp{\mcP}{q^k}v^{k}_{\mcD_{\eps}}\right]\bv^{\mcD_{\eps}}- [D^{\mcV}\bh_{\eps}][\mcP] [\bGamma(\bq,\bv^{\mcD_{\eps}})]\bv^{\mcD_{\eps}}.
    \end{align} Contrastingly, the reaction forces $\bh_{\eps}$ resulting from the generating equation \eqref{generating_eq} are calculated from the equation
    \begin{align}\notag
        \small-\frac{1}{\eps}\mbF\mcR^{\sharp_{\mbG}}(\bv) = -(\Conn{}{\dot{\bq}}{\mcP^{\bot}})(\bv) + D^{\mcV}\bh_{\eps}(\Conn{}{\dot{\bq}}{\mcP})(\bv) + \Conn{\mcH}{\dot{\bq}}{\bh_{\eps}(\mcP(\bv))}.
    \end{align} as was observed in equation \eqref{friction_force}. By Remark \ref{covariant_remark} the covariant derivatives in this equation may be expressed in matrix form by 
    \begin{align} \label{cov_fric} \notag
        &\small -\frac{1}{\eps}\mbF\mcR^{\sharp_{\mbG}}(\bv) = -\left[\sum^{n}_{k=1}\pp{\mcP^{\bot}}{q^k}v^{k}\right]\bv - [\bGamma(\bq,\bv)][\mcP^{\bot}]\bv + [\mcP^{\bot}][\bGamma(\bq,\bv)]\bv\\ \notag
        &\small + [D^{\mcV}\bh_{\eps}]\left[\sum^{n}_{k=1}\pp{\mcP}{q^k}v^{k}\right]\bv + D^{\mcV}\bh_{\eps}[\bGamma(\bq,\bv)][\mcP]\bv -[D^{\mcV}\bh_{\eps}][\mcP][\bGamma(\bq,\bv)]\bv\\ 
        &\small + \left[\pp{\bh_{\eps}}{\bq}(\bq,\bv^{\mcD})\right]\bv + [\bGamma(\bq,\bv)]\bh_{\eps}(\bq,\bv^{\mcD})- [D^{\mcV}\bh_{\eps}][\bGamma(\bq,\bv)][\mcP]\bv.
    \end{align} Then, restricting to $\bv = \bv^{\mcD_{\eps}}$, and using $\mcP^{\bot}\circ \bh_{\eps} = \bh_{\eps}$ terms cancel and therefore, these friction forces \eqref{cov_fric} restricted to $\mcD_{\eps}$ are identical to \eqref{classical_invariance_forces}.
\end{proof}

\begin{corollary}[Equivalence of Tangent Conditions]\label{thm_dynamics_tangent} The dynamics defined by \eqref{slow_fast_eom:eq1}-\eqref{slow_fast_eom:eq3} restricted to $\mcD_{\eps}$ are tangent to the manifold $\mcD_{\eps}$.
\end{corollary}
\begin{proof}
    In local coordinates \eqref{slow_fast_eom:eq1}-\eqref{slow_fast_eom:eq3} are equivalent to the equations of motion on $TTQ$ given by \eqref{geodesic_spray}. Following singular perturbation theory, the classical invariance equation \eqref{classical_invariance_equation} ensures that the reaction forces result in dynamics tangent to $\mcD_{\eps}$.
    By Theorem \ref{thm_equiv_forces}, the reaction friction forces described by \eqref{generating_eq} are equivalent to the reaction forces resulting from the classical invariance equation. Therefore, the dynamics \eqref{slow_fast_eom:eq1}-\eqref{slow_fast_eom:eq3} along with the generating equation \eqref{generating_eq} result in dynamics tangent to $\mcD_{\eps}$. 
\end{proof}

\begin{corollary}[Equivalence of Invariance Conditions]
    The classical invariance condition \eqref{classical_invariance_equation} for the dynamics \eqref{slow_fast_eom:eq1}-\eqref{slow_fast_eom:eq3}
is equivalent to the condition 
\begin{align}\label{cov_invariance_condition}
    \Conn{}{\dot{\bq}}{(\mcP^{\bot}(\bv))}\big|_{\mcD_{\eps}} = \Conn{}{\dot{\bq}}{(\bh_{\eps}\circ\mcP(\bv))}\big|_{\mcD_{\eps}}.
\end{align}
\end{corollary}
\begin{proof}
    By Theorem \ref{thm_equiv_forces} and Corollary \ref{thm_dynamics_tangent} both conditions result in the same dynamical equations of motion restricted to the slow manifold $\mcD_{\eps}$. 
\end{proof}
\begin{remark}
    The later covariant derivative invariance condition \eqref{cov_invariance_condition} may be geometrically interpreted as the parallel transport of the condition $\mcP^{\bot}(\bv) = \bh_{\eps}\circ\mcP(\bv)$ along trajectories $\dot{\bq} = \bv \in \mcD_{\eps}$ lying in the slow manifold.
\end{remark}

\subsection{Approximation of Slip Velocities}
In this section, we employ the generating equation \eqref{generating_eq} from Theorem \ref{thm_generating_eq} to compute the zeroth, first and second order approximations of the slow manifold $\bh_{\eps}$. Expand $\bh_{\eps}$ as a formal power series 
\begin{align}
    \bh_{\eps} = \bh^{(0)} + \eps\bh^{(1)} + \eps^{2}\bh^{(2)} +  O(\eps^3).
\end{align} Further expanding the power series allows for higher order approximations of $\mcD_{\eps}$.
\subsubsection{Zeroth Order Approximation of Slip Velocities} Lets consider the zeroth order term of the slow manifold $\bh_{\eps} = \bh^{(0)} + O(\eps)$. Observe that when $\eps = 0$ we have 
$$\mcD_\eps \big |_{\eps = 0} = \mcD_0  = \{\bv \in TQ \, | \, \mcP^{\bot}(\bv) = \bh^{(0)}\}.$$ Since $\bh_{\eps}$ is a perturbation of the nonholonomic distribution, we find that $\bh^{(0)} = \bzero$. Then, to find the zeroth order velocity,  solve the identity $\mcP^{\bot}(\bv) = \bzero$ for $\bv = \bv^{\mcD} \in \mcD$.
\begin{remark}
    The zeroth order approximation of the slow fast velocities are equivalent to the nonholonomic velocities.  
\end{remark}

\subsubsection{First Order Approximation of Slip Velocities} Let us assume that section $\bh_\eps$ is taken to first order, so we expand $\bh_\eps = \eps \bh^{(1)} + O(\eps^2)$. Then according to the generating equation \eqref{generating_eq} to this first order slip velocity, we find 
\begin{align} \label{first_order_slip}
    \bh^{(1)} = \mbQ_{\bq}\circ \left[(\Conn{}{\bv^{\mcD}}{\mcP^{\bot}})(\bv^{\mcD}) - \mcP^{\bot}(\bdd V^{\sharp_{\mbG}})\right].
\end{align}

\begin{remark}
    The first order approximation of the slow fast velocities are proportional to the nonholonomic Lagrange multipliers \eqref{lagrange_mult_nonholonomic}. This is consistent with previous results \cite{eldering2013normally}.
\end{remark}

\subsubsection{Second Order Approximation of Slip Velocities} Lastly, assume that section $\bh_\eps$ is taken to second order, so we expand the formal power series as $\bh_\eps = \eps \bh^{(1)} + \eps^{2}\bh^{(2)} + O(\eps^3)$. Then according to the generating equation \eqref{generating_eq} to this first order slip velocity, we find 
\begin{align} \label{second_order_slip} 
    &\bh^{(2)} = \mbQ_{\bq}\circ \left[(\Conn{}{\bh^{(1)}}{\mcP^{\bot}})(\bv^{\mcD}) + (\Conn{}{\bv^{\mcD}}{\mcP^{\bot}})(\bh^{(1)})\right] \\ \notag
    &\quad \quad \quad -\mbQ_{\bq}\circ \left[D^{\mcV}\bh^{(1)}\left((\Conn{}{\bv^{\mcD}}{\mcP})(\bv^{\mcD}) \right) - D^{\mcV}\bh^{(1)}\,\mcP(\bdd V^{\sharp_{\mbG}}) \right]- \mbQ_{\bq}\circ \left[\Conn{\mcH}{\bv^{\mcD}}{\bh^{(1)}}(\bv^{\mcD}) \right].
\end{align}

\begin{remark}
    Higher order approximations of $\bh_{\eps}$ may be calculated recursively in terms of higher order covariant derivatives $\Conn{}{}{}^{n}$. However, in practice, first and second order approximations suffice. This is due to the fact that the friction forces $- \frac{1}{\eps}\mbF\mcR^{\sharp_{\mbG}}(\bv)$ in \eqref{slow_fast_eom:eq1} restricted to the second order approximation results in a first order force $- \eps\mbF\mcR^{\sharp_{\mbG}}(\bh^{(2)})$ that is a function of zeroth and first-order velocities. 
\end{remark}

\subsection{Approximation of Slip Dynamics}
In this section, we use the approximations of the slip velocities $\bh_{\eps} = \eps\bh^{(1)} + \eps \bh^{(2)} + O(\eps)$ to approximate the zeroth and first order dynamics. Due to the recurrence relation defined by the generating equation \eqref{generating_eq} on $\bh_{\eps}$ we observe that the zeroth order approximation of the dynamics \eqref{slow_fast_eom:eq1}-\eqref{slow_fast_eom:eq3} depends on the first order approximation $\bh^{(1)}$. Furthermore, the first order approximation of the dynamics \eqref{slow_fast_eom:eq1}-\eqref{slow_fast_eom:eq3} depends on the second order approximation $\bh^{(2)}$, and so on. Therefore, the order of the slip velocities $\bh_{\eps}$ is always one order above the order of the dynamics. 

\subsubsection{Zeroth Order Dynamics} By the equations of motion \eqref{slow_fast_eom:eq1} and the first order approximation of the slow manifold $\bh_{\eps} = \eps\bh^{(1)} + O(\eps^2)$, we find the zeroth order equations of motion are given by
\begin{align}\label{zeroth_eom}
    &\dd{\bv}{t} = -\boldsymbol{\Gamma}(\bq,\bv^{\mcD})\bv^{\mcD} - (\Conn{}{\bv^{\mcD}}{\mcP^{\bot}})(\bv^{\mcD}) - \mcP(\bdd V^{\sharp_{\mbG}})\\
    &\dd{\bq}{t}  = \bv^{\mcD} \in \mcD_{0}\cong \mcD.
\end{align}

\begin{remark}
    The zeroth order equations of motion \eqref{zeroth_eom} are exactly the nonholonomic equations of motion \eqref{nonholonomic_affine:eq1}-\eqref{nonholonomic_affine:eq3}. This result is consistent with previous realizations of slow/fast nonholonomic systems \cite{eldering2016realizing}. 
\end{remark}

\subsubsection{First Order Dynamics} Using the generating function \eqref{generating_eq} we obtain the first and second order approximations $\bh_{\eps} = \eps\bh^{(1)} + \eps^{2}\bh^{(2)} + O(\eps^3)$ to obtain \eqref{first_order_slip} and \eqref{second_order_slip}, the former gives the slip velocity and the later gives the reaction force. Restricting the dynamics to the first order approximation of slip, we find the dynamics: 

\begin{align}\label{first_eom}
    &\begin{aligned}
        \dd{\bv}{t}&= 
    -\boldsymbol{\Gamma}(\bq,\bv^{\mcD})\bv^{\mcD} - (\Conn{}{\bv^{\mcD}}{\mcP^{\bot}})(\bv^{\mcD}) - \mcP(\bdd V^{\sharp_{\mbG}})\\ \notag
    &\quad \small -\eps\left[\boldsymbol{\Gamma}(\bq,\bv^{\mcD})\bh^{(1)} +\boldsymbol{\Gamma}(\bq,\bh^{(1)})\bv^{\mcD} \right] +\eps\Conn{\mcH}{\bv^{\mcD}}{\bh^{(1)}}(\bv^{\mcD})\\ \notag
    &\quad -\eps \left[(\Conn{}{\bv^{\mcD}}{\mcP^{\bot}})(\bh^{(1)}) + (\Conn{}{\bh^{(1)}}{\mcP^{\bot}})(\bv^{\mcD}) \right]\\ 
    &\quad  + \eps \, D^{\mcV}\bh^{(1)}\left[\left((\Conn{}{\bv^{\mcD}}{\mcP})(\bv^{\mcD}) \right) - \mcP(\bdd V^{\sharp_{\mbG}}) \right]
    \end{aligned}\\
    &\dd{\bq}{t} = \bv^{\mcD} + \eps \bh^{(1)}(\bv^{\mcD})
\end{align} where the first order approximation of the slip $\bh^{(1)}$ is given by \eqref{first_order_slip}, and vertical derivative $D^{\mcV}\bh^{(1)}$ and horizontal $\Conn{\mcH}{\bv^{\mcD}}{\bh^{(1)}}$ is calculated according to equation \eqref{vertical_deriv} and \eqref{horizontal_deriv}. 

    



\section{Case Study: Vertical Rolling Disk} \label{sec_case_study}
{To illustrate the theory, we consider a vertical rolling disk for simplicity. However, our approach may be applied to other systems with (linear) nonholonomic rolling constraints such as wheeled mobile robots \cite{samson2016modeling} and vehicle systems \cite{chhabra2016dynamical}. For the traditional derivation of the nonholonomic equations for the vertical rolling disk using the Lagrange-d'Alembert variational principle see \cite[Sec. 1.4]{bloch2004nonholonomic} and for the affine-connection approach see \cite[Sec. 4.5]{bullo2019geometric}, and \cite{lewis2000simple}. 

The configuration space for a planar vertical rolling disk is $Q = SE(2)\times SO(2)$ with local coordinates $\bq = (\theta,x,y,\varphi)$ and local coordinates for velocities $\bv = (v_{\theta},v_{x},v_{y},v_{\varphi})$. The Riemannian metric defining the kinetic energy is given by 
\begin{align}\label{wheel_metric}
        \mbG_{\bq} = \text{diag}\left(I,m,m,J \right)
\end{align} where $m$ is the mass of the disk, $I$ is the moment of inertia normal to the plane and $J$ is the moment of inertia about the rolling axis. Note that as $\mbG_{\bq}$ is constant, the corresponding Christoffel symbols $\bGamma(\bq,\bv) = \bzero$ for all $(\bq,\bv) \in TQ$ \ie, the metric $\mbG_{\bq}$ is flat. 

The nonholonomic constraint one-forms 
\begin{align} \label{wheel_constriant_forms}
    a_{q}^{1} = \bdd x - R\cos \theta \bdd \varphi,\quad a_{q}^{2} = \bdd y - R\sin \theta \bdd \varphi,
\end{align} enforcing rolling without slipping are given by 
\begin{align}\label{wheel_constriants}
    &a_{q}^{1}(v_{\theta},v_{x},v_{y},v_{\varphi}) = v_x - R\cos \theta v_{\varphi} = 0,\\
    &a_{q}^{2}(v_{\theta},v_{x},v_{y},v_{\varphi}) = v_y - R\sin \theta v_{\varphi} = 0,
\end{align} where $R$ is the radius of the disk. This constraint defines a nonholonomic distribution $\mcD_{\bq} = \ker(a^{1}_{q})\cap\ker(a^{2}_{q})$ which are spanned by the vector fields 
\begin{align}\label{vf_distribution}
    \bpp{}{\theta}, \quad R\cos\theta\bpp{}{x} + R\sin\theta \bpp{}{y} + \bpp{}{\varphi}.
\end{align} Using the musical isomorphisms, we find that the $\mbG$-orthogonal complement spanned by the orthogonal complement $\mcD^{\bot} = \left(\mcD^{\circ} \right)^{\flat_{\mbG}}$ is spanned by the vector fields 
\begin{align}\label{vf_distribution_perp}
    \frac{1}{m}\bpp{}{x} -\frac{R}{J}\cos\theta\bpp{}{\varphi}, \quad \frac{1}{m}\bpp{}{y} - \frac{R}{J}\sin\theta \bpp{}{\varphi}. 
\end{align} The projection map $\mcP:TQ \rightarrow \mcD $ onto the distribution $\mcD$ is constructed by the vector fields spanned by \eqref{vf_distribution} and the metric \eqref{wheel_metric} and is given by 
\begin{align}
    \mcP = \mcI\begin{bmatrix}
        \frac{1}{\mcI}& 0 & 0 & 0 \\
        0 & \cos^2\theta & \frac{1}{2}\sin2\theta & \frac{J}{mR}\cos\theta\\
        0 & \frac{1}{2}\sin2\theta & \sin^2\theta & \frac{J}{mR}\sin\theta \\
        0 & \frac{\cos\theta}{R} & \frac{\sin\theta}{R} & \frac{J}{mR^2}
    \end{bmatrix}
\end{align} where $\mcI:= \frac{mR^2}{J+mR^2}$ is a ratio of inertia. 

Similarly, the projection map $\mcP^{\bot}:TQ \rightarrow \mcD^{\bot} $ onto the distribution $\mcD^{\bot}$ is constructed by the vector fields spanned by \eqref{vf_distribution} and the metric \eqref{wheel_metric} and is given by 
\begin{align}
    \small \mcP^{\bot} = \mcI\begin{bmatrix}
        0& 0 & 0 & 0 \\
        0& \frac{J + mR^2\sin^2\theta}{mR^2} & -\frac{1}{2}\sin2\theta & -\frac{J}{mR}\cos\theta \\
        0& -\frac{1}{2}\sin2\theta & \frac{J + mR^2\cos^2\theta}{mR^2} & -\frac{J}{mR}\sin\theta \\
        0& -\frac{\cos\theta}{R} & -\frac{\sin\theta}{R}& 1
    \end{bmatrix}.
\end{align} The nonholonomic constraints may then be expressed by the equation $\mcP^{\bot}(\bv) = \bzero$, from which we find the constrained velocities $\bv^{\mcD} \in \mcD$ is given by 
\begin{align}
    \bv^\mcD = \begin{bmatrix}
        v_{\theta} & R\cos\theta v_{\varphi} & R\sin\theta v_{\varphi} & v_{\varphi}
    \end{bmatrix}^{T}
\end{align} in the coordinate induced frame $\left\{\bpp{}{\theta},\bpp{}{x},\bpp{}{y},\bpp{}{\varphi} \right\}$.

We construct the Rayleigh dissipation function according to equations \eqref{Rayleigh_force} using the nonholonomic constraint form \eqref{wheel_constriants}, from which the coefficient matrix is given by 
\begin{align}\label{rw_friction_matrix}
    \mbF\mcR^{\sharp_{\mbG}} = \begin{bmatrix}
        0& 0& 0 & 0 \\
        0 & \frac{\mu}{m} & 0 & -\frac{\mu R}{m}\cos\theta \\
        0 & 0 & \frac{\mu}{m} & -\frac{\mu R}{m}\sin\theta\\
        0 & -\frac{\mu R}{J}\cos\theta& -\frac{\mu R}{J}\sin\theta& \frac{\mu R^2}{J}
    \end{bmatrix}
\end{align} where the units of the friction coefficients $\mu$ are $[{kg}/{s} ]$, and the corresponding accelerations given by $\bF_{\eps} = -\frac{1}{\eps}\mbF\mcR^{\sharp_{\mbG}}\bv$. Defining a frame $\Phi$ as the direct sum of the vector fields \eqref{vf_distribution} and \eqref{vf_distribution_perp} we find the matrix $\mbQ_{\bq}$ according to Remark \ref{Q_remark} is given by 
\begin{align}
    \small\mbQ_{\bq} = \begin{bmatrix}
        0 & 0 & 0 & 0\\
        0 & J^2 + \gamma m R^2 \sin^2\theta & -\frac{1}{2}\gamma mR^2\sin2\theta & J^2R\cos\theta \\
        0 &  -\frac{1}{2}\gamma mR^2\sin2\theta & J^2 + \gamma m R^2 \cos^2\theta & - J^2 R\sin\theta \\
        0& -mRJ\cos\theta & -mRJ\sin\theta & mR^2J
    \end{bmatrix}
\end{align} where $\gamma = mR^2 + 2J$. One can check that $\mbQ_{\bq}$ satisfies Lemma \ref{prop_first_isomorphism}. As the Christoffel symbols are all zero, the equations of motion are simply given by 

\begin{align}\label{vd_friction_eom:eq1}
    &\dd{}{t}\begin{bmatrix}
        v_{\theta} \\ v_{x}\\ v_{y}\\ v_{\varphi}
    \end{bmatrix}  = -\frac{1}{\eps}\begin{bmatrix}
        0 \\ \frac{\mu}{m}v_{v} - \frac{\mu \rho}{m}\cos\theta v_{\varphi}\\
        \frac{\mu}{m}v_{y} - \frac{\mu \rho}{m}\sin\theta v_{\varphi}\\ 
        -\frac{\mu\rho}{J}\cos\theta v_{x} - \frac{\mu\rho}{J}\sin\theta v_{y} + \frac{\mu\rho^2}{J}v_{\varphi}
    \end{bmatrix}, \quad \\  \label{vd_friction_eom:eq2}
    & \dd{}{t}\begin{bmatrix}
        \theta \\ x \\ y \\ \varphi
    \end{bmatrix} = \begin{bmatrix}
        v_{\theta} \\ v_{x}\\ v_{y}\\ v_{\varphi}
    \end{bmatrix}.
\end{align}
\begin{remark}
    Note that the dynamics of the orientation in the plane $\theta$ is of the form $\dot{v}_{\theta} = 0$ as $v_\theta \in \mcD$ and $\mcD = \ker\left(\mbF\mcR^{\sharp_{\mbG}}\right)$, by construction. Hence, the orientation in the plane is invariant for the dynamics of the vertical rolling disk. 
\end{remark}

\subsection{Approximations of Slip Velocities} The first order approximation of the slip velocities are calculated according to  \eqref{first_order_slip} and are found as 
\begin{align} \label{wheelslip_fo_approx}
    \bh^{(1)} = \mbQ_{\bq}\circ(\Conn{}{\bv^{\mcD}}{\mcP^{\bot}})(\bv^{\mcD}) =  \begin{bmatrix}
    0 \\ \frac{mR}{\mu}\sin\theta v_{\theta}v_{\varphi}\\-\frac{mR}{\mu}\cos\theta v_{\theta}v_{\varphi}\\ 0 
    \end{bmatrix}
\end{align} Further, the second order approximation of the slip velocities is calculated according to \eqref{second_order_slip} using the results of Remark \ref{covariant_remark} and is found to be 
\begin{align}  \label{wheelslip_so_approx} \notag
    \bh^{(2)} &=     \mbQ_{\bq}\circ \left[(\Conn{}{\bh^{(1)}}{\mcP^{\bot}})(\bv^{\mcD}) + (\Conn{}{\bv^{\mcD}}{\mcP^{\bot}})(\bh^{(1)})\right] \\ \notag
    &\quad \quad \quad -\mbQ_{\bq}\circ \left[D^{\mcV}\bh^{(1)}\left((\Conn{}{\bv^{\mcD}}{\mcP})(\bv^{\mcD}) \right) \right] - \mbQ_{\bq}\circ \left[\Conn{\mcH}{\bv^{\mcD}}{\bh^{(1)}}(\bv^{\mcD}) \right] \\ 
    &=\frac{m^2RJ^2}{\mu^2(J + mR^2)^2}\begin{bmatrix}
        0 \\ \cos\theta v_{\theta}^2 v_{\varphi}\\ \sin\theta v_{\theta}^2 v_{\varphi} \\ -\frac{mR}{J}v_{\theta}^2 v_{\varphi}
    \end{bmatrix}
\end{align}

\subsection{Approximations of Slip Dynamics} Using the first order approximation of the slip velocities \eqref{wheelslip_fo_approx}, we find the zeroth order approximation of \eqref{vd_friction_eom:eq1}-\eqref{vd_friction_eom:eq2} using equations \eqref{zeroth_eom} are given by 
\begin{align}\label{vd_zeroth_eom:eq1}
    &\dd{}{t}\begin{bmatrix}
        v_{\theta} \\ v_{x}\\ v_{y}\\ v_{\varphi}
    \end{bmatrix} = \begin{bmatrix}
        0 \\ -R\sin\theta v_{\theta} v_{\varphi}\\ R\cos\theta v_{\theta} v_{\phi}\\ 0
    \end{bmatrix} +O(\eps), \\  \label{vd_zeroth_eom:eq2}
    &\dd{}{t}\begin{bmatrix} 
        \theta \\ x \\ y \\ \varphi
    \end{bmatrix} = \begin{bmatrix} 
        v_{\theta} \\ R\cos\theta v_{\varphi} \\ R\sin\theta v_{\phi} \\ v_{\varphi}
    \end{bmatrix} +O(\eps). 
\end{align} 

\begin{remark}
    The zeroth order equations of motion \eqref{vd_zeroth_eom:eq1}-\eqref{vd_zeroth_eom:eq2} are the same equations of motion resulting from Lagrange-d'Alembert's variational principle \cite[Sec. 1.4]{bloch2004nonholonomic}. 
\end{remark}
Moreover, using the second order approximation of the slip velocities \eqref{wheelslip_fo_approx} and \eqref{wheelslip_so_approx}, we find the first order equations \eqref{first_eom} are given by 
\begin{align}
    \small &\dd{}{t}\begin{bmatrix}
        v_{\theta} \\ v_{x}\\ v_{y}\\ v_{\varphi}
    \end{bmatrix} = \begin{bmatrix}
        0 \\ -R\sin\theta v_{\theta} v_{\varphi}\\ R\cos\theta v_{\theta} v_{\varphi}\\ 0
    \end{bmatrix} +\eps \small{\frac{mRJ}{\mu(J+mR^2)} }\begin{bmatrix}
        0 \\ -\cos\theta v_{\varphi}v_{\theta}^{2}\\-\sin\theta v_{\varphi}v_{\theta}^{2} \\ \frac{mR}{J} v_{\varphi}v_{\theta}^{2}
    \end{bmatrix}+O(\eps^2),\\
    \small &\dd{}{t}\begin{bmatrix}
        \theta \\ x \\ y \\ \varphi
    \end{bmatrix} = \begin{bmatrix}
        v_{\theta} \\ R\cos\theta v_{\varphi} \\ R\sin\theta v_{\varphi}\\ v_{\varphi}
    \end{bmatrix} + \eps \frac{mR}{\mu}\begin{bmatrix}
    0 \\ \sin\theta v_{\theta}v_{\varphi}\\-\cos\theta v_{\theta}v_{\varphi}\\ 0 
    \end{bmatrix} +O(\eps^2).
\end{align}

\section{Conclusion}
In this paper, we developed an affine connection approach for realizing nonholonomic mechanical systems. We formulated the slow/fast dynamics in a coordinate-independent manner in terms of sections, projection maps and covariant derivatives. We proposed a novel invariance condition for slow/fast realizations of nonholonomic constraints and proved that this novel condition is equivalent to the classical approach. We proposed a recursive procedure to approximate higher order slip velocities and their associated dynamics. We conclude the paper with some remarks regarding the advantages of the proposed formalism. 
\begin{enumerate}
    \item The geometric formalism is coordinate-independent and suitable for symbolic computation.
    \item The proposed slow-fast decomposition permits a recursive procedure to approximate the invariant manifold representing slip velocities. In the context of vehicle systems, the slow manifold may be physically interpreted as a deviation from the ideal no-slip constraints. Approximations of the slow manifold provide estimates of these deviations and have applications in model-based slip estimation and localization in vehicle systems. 
    \item Higher order approximations of slip velocities are based on objects that can be interpreted geometrically and may give insight into the geometric structure of other equivalent approaches. 
\end{enumerate} The future directions of this research include the following: (i) Designing a recursive controller robust to the uncertainties in the estimation of the friction coefficients and external forces; (ii) Use of higher-order approximations of the slow-manifold for in-the-loop motion estimation of vehicle systems; and (iii) symmetry reduction for slow-fast realizations of nonholonomic systems. 
\\


\textbf{Conflict of Interest:}  The authors declare that they have no conflict of interest.

\begin{appendices}

\section{Proofs of Statements}

\subsection{Proof of Proposition \ref{prop_projection}} \label{proof_prop_projection}
\begin{proof} 
    For the horizontal space  
\begin{align} \notag
    \mcP_{\mcH}\left(\bq,\bv_{\bq},\bU,\bW \right) &= \left(\mcQ_{\Gamma}^{\mcH}\left(\bq,\bv_{\bq},\bU,\bW \right)\right)^{\down{\hor}}\\ \notag
    &= \left(\bq,\bv_{\bq},\bU,- \Gamma(\bq,\bv_{\bq})\bU \right)^{\down{\hor}}\\ \notag
    &= \left(\bq,\bU \right) = \left(\pi_{TQ} \right)_{*}(\bq,\bv_{\bq},\bU,\bW). \notag
\end{align} Therefore we have $\mcP_{\mcH} = \left(\pi_{TQ} \right)_{*}$. 

For the vertical space, let $\bX:Q \rightarrow TQ$ be a section:
\begin{align} \notag
    \mcP_{\mcV}(\bX_{*}) &= \left(\mcQ^{\mcV}_{\Gamma}\left(\bq,\bX_{\bq},\bv_{\bq},D\bX(\bq)\cdot \bv_{\bq} \right) \right)^\down{\ver}\\ \notag
    &= \left(\bq,\bX_{\bq},\bzero, D\bX(\bq)\cdot \bv_{\bq} + \Gamma(\bq,\bX_{\bq})\cdot \bv_{\bq} \right)^\down{\ver}\\ \notag
    &= \left(\bq,D\bX(\bq)\cdot \bv_{\bq} + \Gamma(\bq,\bX_{\bq})\cdot \bv_{\bq} \right)
    = \Conn{}{\bv}{\bX}. \notag
\end{align}

\end{proof}

\subsection{Proof of Proposition \ref{prop_bundle_map}} \label{proof_prop_bundle_map}

\begin{proof}
Let $\bh:TQ \rightarrow TQ$ be a vector bundle map, i.e., $\pi_{TQ}\circ \bh = \pi_{TQ}$. Let $\bw,\,\bv:Q\rightarrow TQ$ be sections, then $\bh\circ\bw : Q \rightarrow TQ$ is also a section. The push-forward of this section is a map that satisfies the following commutative diagram: 
\begin{center}
\[\begin{tikzcd}
	TTQ && TTQ \\
	TQ && TQ \\
	& Q
	\arrow["h", from=2-1, to=2-3]
	\arrow["{\pi_{TQ}}", from=2-3, to=3-2]
	\arrow["{\pi_{TQ}}"', from=2-1, to=3-2]
	\arrow["{\boldsymbol{w}}", curve={height=-24pt}, from=3-2, to=2-1]
	\arrow["{\pi_{TTQ}}", from=1-3, to=2-3]
	\arrow["{h_{*}}", from=1-1, to=1-3]
	\arrow["{\boldsymbol{w}_{*}}", curve={height=-12pt}, from=2-1, to=1-1]
	\arrow["{\pi_{TTQ}}", curve={height=-6pt}, from=1-1, to=2-1]
\end{tikzcd}\]
\end{center} By virtue of Proposition \eqref{prop_projection}, we calculate the covariant derivative along $\bv$ as follows
\begin{align} \notag
    \Conn{}{\bv}{(\bh\circ \bw)} &= \mcP_{\mcV}\left(D(\bh\circ \bw) \right) = \mcP_{\mcV}\left(D\bh\cdot D\bw \right)\\ \notag
    &= \mcP_{\mcV}\left(D\bh\cdot (\mcP_{\mcH}\oplus\mcP_{\mcV})(D\bw)  \right)\\ \notag
    &= \mcP_{\mcV}\left(D\bh\cdot \left(\mcP_{\mcH}(D\bw)+\mcP_{\mcV}(D\bw)  \right)\right)\\ \notag
    \small&= \mcP_{\mcV}\left(D\bh\cdot\left(\mcP_{\mcH}(D\bw)\oplus\bzero \right) \right) + \mcP_{\mcV}\left(D\bh\cdot(\bzero\oplus\Conn{}{\bv}{\bw}) \right), \notag
\end{align} where we used the direct product isomorphism, twice and $D\bh$ is the Jacobian of $\bh$. In the decomposition $TTQ \cong  \mcH TQ\oplus \mcV TQ$, consider the Jacobian of the map $\bh$:
\begin{align}
    D\bh = \begin{bmatrix}\mcP_{\mcH}D\bh\big|_{\mcH}& \mcP_{\mcH}D\bh\big|_{\mcV}\\
    \mcP_{\mcV}D\bh\big|_{\mcH}& \mcP_{\mcV}D\bh\big|_{\mcV}\end{bmatrix}.
\end{align} Since $\mcP_{\mcH} = \left(\pi_{TQ}\right)_{*}$,  when restrict $D\bh$ to the vertical space $\mcV TQ = \ker(\pi_{TQ})_{*}$ we find 
\begin{align}
    \pi_{\mcH}D\bh\big|_{\mcV} = \left(\pi_{TQ}\right)_{*}\bh_{*}\big|_{\mcV}  = \left(\pi_{TQ}\right)_{*}\big|_{\mcV}  = \bzero.
\end{align} Let $\bz:[a,b] \rightarrow TQ$ be a horizontal curve along $\bq(t)$ in $Q$. Then, $\dot{\bz} \in T_{\bw(t)}TQ \cong \mcH_{\bz(t)}TQ\oplus \mcV_{\bz(t)}TQ $ satisfying $$\dot{\bz}  = \dot{\bq}^{i}\left(\bpp{}{\tl{q}^{i}} - \Gamma^{k}_{ij}v^{j}\bpp{}{v^{i}} \right).$$ Since $\bh(\bz(t)) \in T_{\bq(t)}Q$ is a curve, we calculate derivative 
\begin{align} \notag
    \dot{\bq}(t) &= \dd{}{t}\left(\pi_{TQ}\circ \bh\circ \bz(t)  \right) = (\pi_{TQ})_{*}\, D\bh\, \dot{\bz} = \mcP_{\mcH}(D\bh\, \dot{\bz}).\notag
\end{align} Therefore, we find that $\mcP_{\mcH}D\bh\big|_{\mcH}$ is the identity map $\text{id}$. 

Define the vertical component of the Jacobian by $D^{\mcV}\bh:= \mcP_{\mcV}\left(D\bh \right)|_{\mcV}$ and the horizontal covariant derivative $\nabla^{\mcH}\bh:= \mcP_{\mcV}D\bh\big|_{\mcH}$. 
Therefore, the covariant derivative of the bundle map $\bh$ is given by 
\begin{align} \notag
    \Conn{}{\bv}{(\bh\circ \bw)} = \Conn{\mcH}{\bv}{\bh(\bw)} + D^{\mcV}\bh(\Conn{}{\bv}{\bw}).
\end{align}
We calculate local coordinate expressions for the covariant derivative as follows:
\begin{align} \notag
    \Conn{}{\bv}{\left(\bh\circ \bw \right)} &= v^i\Conn{}{\bpp{}{q^i}}{\left(\left(\bh\circ \bw\right)^j\bpp{}{q^j}\right)}\\ \notag
    &= v^i \left(\mcL_{\bpp{}{q^i}}(\bh(\bq,\bw)^j)\bpp{}{q^j} + \bh(\bq,\bw)^j\Conn{}{\bpp{}{q^i}}{\bpp{}{q^j}} \right)\\ \notag
    \small&= v^i \left(\pp{\bh^j}{q^i}\bpp{}{q^j} + \pp{\bh^j}{w^k}\pp{w^k}{q^i}\bpp{}{q^j} + \bh^i\Gamma^k_{ij}\bpp{}{q^k} \right). \notag
\end{align} On the other hand we have that 
\begin{align}
    \Conn{}{\bpp{}{q^i}}{\bw} = \left(\pp{w^k}{q^i} + \Gamma_{ij}w^j\right)\bpp{}{q^k}. \notag
\end{align} Putting this expression into the above formula we find 
\begin{align} \notag
    \Conn{}{\bv}{\left(\bh\circ \bw \right)} &= \pp{\bh^m}{w^l}\left(\Conn{}{\bv}{\bw} \right)^l \bpp{}{q^m} + \left(\pp{\bh^m}{q^k}v^k  + \Gamma^{m}_{ki}v^{k}\bh^{i} - \Gamma^{l}_{ki}w^{k}v^{i}\pp{\bh^m}{w^l} \right)\bpp{}{q^m}\\ \notag 
    &=  D^{\mcV}\bh(\Conn{}{\bv}{\bw}) + \Conn{\mcH}{\bv}{\bh(\bw)}. \notag
\end{align} 
\end{proof}

\end{appendices}


\bibliography{sn-article}

\end{document}